\pdfoutput=1
\documentclass{article}
\usepackage{amsmath}
\usepackage{amsthm}
\usepackage{mathptmx}
\usepackage[libertine]{newtxmath}
\usepackage[T1]{fontenc}
\usepackage{geometry}
\usepackage{color}
\usepackage{array}
\usepackage{pifont}
\usepackage{float}
\usepackage{mathtools}
\PassOptionsToPackage{normalem}{ulem}
\usepackage{ulem}
\usepackage{nameref}

\makeatletter

\providecommand{\tabularnewline}{\\}

\usepackage{etoolbox}
\usepackage[bottom,hang]{footmisc}
\usepackage[medium,compact]{titlesec}	
\usepackage{xcolor}
\usepackage{hyperref}	
\hypersetup
{
  unicode=true,
  colorlinks=true,
  linkcolor=[rgb]{0 0 .8},
  citecolor=[rgb]{0 .8 0},
  urlcolor=[rgb]{.8 0 .8},
  breaklinks=true,  
  pdfstartview=,
  bookmarksopen=true,
}
\usepackage[numbered]{bookmark}
\usepackage[all]{hypcap}
\usepackage[nameinlink]{cleveref}	

\AtBeginDocument{\let\ref\cref}	
\crefname{section}{Section}{Sections}
\crefname{figure}{Figure}{Figures}
\crefname{table}{Table}{Tables}
\crefname{example}{Example}{Examples}
\crefname{footnote}{}{}
\crefformat{footnote}{#2\footnotemark[#1]#3}
\makeatletter
\patchcmd\@combinedblfloats
  {\box\@outputbox}
  {\unvbox\@outputbox}
  {}{\errmessage{\noexpand\@combinedblfloats could not be patched}}
\makeatother
\makeatletter\g@addto@macro{\UrlBreaks}{\UrlOrds}\makeatother	
\usepackage{graphicx}
\let\oldsqrt\sqrt
\renewcommand{\sqrt}[2][\ \,\,]{{\!\!\oldsqrt[\raisebox{.1em}{\scalebox{.7}{$#1$}}]{#2}\,}}
  \let\originalleft\left
  \let\originalright\right
  \renewcommand{\left}{\mathopen{}\mathclose\bgroup\originalleft}
  \renewcommand{\right}{\aftergroup\egroup\originalright}
\binoppenalty=\maxdimen
\DeclareMathSizes{10}{10}{8.5}{7}
\DeclareMathSizes{10.95}{10.95}{9}{7}	
\DeclareMathSizes{12}{12}{10}{8}
\makeatletter
\AtBeginDocument
{
  \check@mathfonts
  \fontdimen13\textfont2 = \fontdimen13\textfont2		
  \fontdimen14\textfont2 = \fontdimen13\textfont2		
  \fontdimen15\textfont2 = \fontdimen13\textfont2		
  \fontdimen16\textfont2 = 1.3\fontdimen16\textfont2	
  \fontdimen17\textfont2 = \fontdimen16\textfont2		
  \fontdimen14\scriptfont2 = \fontdimen13\scriptfont2	
  \fontdimen15\scriptfont2 = \fontdimen13\scriptfont2
  \fontdimen16\scriptfont2 = 1.3\fontdimen16\scriptfont2
  \fontdimen17\scriptfont2 = \fontdimen16\scriptfont2
  \fontdimen14\scriptscriptfont2 = \fontdimen13\scriptscriptfont2	
  \fontdimen15\scriptscriptfont2 = \fontdimen13\scriptscriptfont2
  \fontdimen16\scriptscriptfont2 = 1.3\fontdimen16\scriptscriptfont2
  \fontdimen17\scriptscriptfont2 = \fontdimen16\scriptscriptfont2
}
\makeatother
\usepackage[skip=.5em]{caption}
\newcommand{\setmulength}[2]{#1=#2\relax}	
\setmulength{\thinmuskip}{2mu plus 1mu minus 1mu}
\setmulength{\medmuskip}{2mu plus 1mu minus 1mu}
\setmulength{\thickmuskip}{4mu plus 1mu minus 2mu}
\AtBeginEnvironment{pmatrix}
{
  \setlength{\arraycolsep}{.4em}
  
}
\usepackage{microtype}
\DisableLigatures{encoding=OT1,family=ntxtlf}
\DeclareSymbolFont{unspacedletters}{OT1}{ntxtlf}{m}{it}
\SetSymbolFont{unspacedletters}{bold}{OT1}{ntxtlf}{b}{it}
\DeclareMathSymbol{A}{\mathalpha}{unspacedletters}{`A}
\DeclareMathSymbol{B}{\mathalpha}{unspacedletters}{`B}
\DeclareMathSymbol{C}{\mathalpha}{unspacedletters}{`C}
\DeclareMathSymbol{D}{\mathalpha}{unspacedletters}{`D}
\DeclareMathSymbol{E}{\mathalpha}{unspacedletters}{`E}
\DeclareMathSymbol{F}{\mathalpha}{unspacedletters}{`F}
\DeclareMathSymbol{G}{\mathalpha}{unspacedletters}{`G}
\DeclareMathSymbol{H}{\mathalpha}{unspacedletters}{`H}
\DeclareMathSymbol{I}{\mathalpha}{unspacedletters}{`I}
\DeclareMathSymbol{J}{\mathalpha}{unspacedletters}{`J}
\DeclareMathSymbol{K}{\mathalpha}{unspacedletters}{`K}
\DeclareMathSymbol{L}{\mathalpha}{unspacedletters}{`L}
\DeclareMathSymbol{M}{\mathalpha}{unspacedletters}{`M}
\DeclareMathSymbol{N}{\mathalpha}{unspacedletters}{`N}
\DeclareMathSymbol{O}{\mathalpha}{unspacedletters}{`O}
\DeclareMathSymbol{P}{\mathalpha}{unspacedletters}{`P}
\DeclareMathSymbol{Q}{\mathalpha}{unspacedletters}{`Q}
\DeclareMathSymbol{R}{\mathalpha}{unspacedletters}{`R}
\DeclareMathSymbol{S}{\mathalpha}{unspacedletters}{`S}
\DeclareMathSymbol{T}{\mathalpha}{unspacedletters}{`T}
\DeclareMathSymbol{U}{\mathalpha}{unspacedletters}{`U}
\DeclareMathSymbol{V}{\mathalpha}{unspacedletters}{`V}
\DeclareMathSymbol{W}{\mathalpha}{unspacedletters}{`W}
\DeclareMathSymbol{X}{\mathalpha}{unspacedletters}{`X}
\DeclareMathSymbol{Y}{\mathalpha}{unspacedletters}{`Y}
\DeclareMathSymbol{Z}{\mathalpha}{unspacedletters}{`Z}
\DeclareMathSymbol{a}{\mathalpha}{unspacedletters}{`a}
\DeclareMathSymbol{b}{\mathalpha}{unspacedletters}{`b}
\DeclareMathSymbol{c}{\mathalpha}{unspacedletters}{`c}
\DeclareMathSymbol{d}{\mathalpha}{unspacedletters}{`d}
\DeclareMathSymbol{e}{\mathalpha}{unspacedletters}{`e}
\DeclareMathSymbol{f}{\mathalpha}{unspacedletters}{`f}
\DeclareMathSymbol{g}{\mathalpha}{unspacedletters}{`g}
\DeclareMathSymbol{h}{\mathalpha}{unspacedletters}{`h}
\DeclareMathSymbol{i}{\mathalpha}{unspacedletters}{`i}
\DeclareMathSymbol{j}{\mathalpha}{unspacedletters}{`j}
\DeclareMathSymbol{k}{\mathalpha}{unspacedletters}{`k}
\DeclareMathSymbol{l}{\mathalpha}{unspacedletters}{`l}
\DeclareMathSymbol{m}{\mathalpha}{unspacedletters}{`m}
\DeclareMathSymbol{n}{\mathalpha}{unspacedletters}{`n}
\DeclareMathSymbol{o}{\mathalpha}{unspacedletters}{`o}
\DeclareMathSymbol{p}{\mathalpha}{unspacedletters}{`p}
\DeclareMathSymbol{q}{\mathalpha}{unspacedletters}{`q}
\DeclareMathSymbol{r}{\mathalpha}{unspacedletters}{`r}
\DeclareMathSymbol{s}{\mathalpha}{unspacedletters}{`s}
\DeclareMathSymbol{t}{\mathalpha}{unspacedletters}{`t}
\DeclareMathSymbol{u}{\mathalpha}{unspacedletters}{`u}
\DeclareMathSymbol{v}{\mathalpha}{unspacedletters}{`v}
\DeclareMathSymbol{w}{\mathalpha}{unspacedletters}{`w}
\DeclareMathSymbol{x}{\mathalpha}{unspacedletters}{`x}
\DeclareMathSymbol{y}{\mathalpha}{unspacedletters}{`y}
\DeclareMathSymbol{z}{\mathalpha}{unspacedletters}{`z}
\pretocmd{\section}{\addvspace{0em plus 1.5em}\penalty-2000}{}{}
\pretocmd{\subsection}{\addvspace{0em plus 1em}\penalty-1500}{}{}
\pretocmd{\subsubsection}{\addvspace{0em plus 1em}\penalty-1000}{}{}
  \usepackage{enumitem}
  \setlistdepth{14}
  \makeatletter
  \renewcommand\p@enumii{\theenumi}
  \renewcommand\p@enumiii{\theenumi\theenumii}
  \renewcommand\p@enumiv{\theenumi\theenumii\theenumiii}
  \makeatother
  \newlength{\listspace}
  \setlist{topsep=\listspace,itemsep=\listspace,parsep=0em,partopsep=0em}
  \setlist[enumerate]{leftmargin=2em}
  \setlist[itemize]{leftmargin=1.5em}
  \AtBeginEnvironment{itemize}{\interlinepenalty=3000}
  \AtBeginEnvironment{enumerate}{\interlinepenalty=3000}
  \AfterEndEnvironment{itemize}{\par\penalty0}
  \AfterEndEnvironment{enumerate}{\par\penalty0}
\makeatletter\newcommand{\justified}{\rightskip\z@skip\leftskip\z@skip}\makeatother
\date{}
\usepackage{upgreek}
\usepackage{calc}	
\newlength{\linespace}
\newcommand\makelinespace{\setlength{\linespace}{\baselineskip-1em}\vspace{\linespace}}
\newlength{\parspace}
\newcommand\makeparspace{\setlength{\parspace}{\parskip+\baselineskip-1em}\vspace{\parspace}}
\newcommand{\forcefontspace}{\par\vspace{-\baselineskip}\vphantom{ABCDEgjpqy}}	
\let\incgraphics\includegraphics
\newsavebox{\imagebox}
\newlength{\imagerule}
\newcommand{\imagescale}{1}

\renewcommand{\includegraphics}[2][]
{%
\def\image{\scalebox{\imagescale}{\incgraphics[#1]{#2}}}%
\savebox{\imagebox}{\image}%
\setlength{\imagerule}{\ht\imagebox+\baselineskip-.7em}%
\ifvmode{\forcefontspace}\fi\rule[0em]{0em}{\imagerule}\image%
}
\let\oldfigure\figure
\let\oldtable\table
\makeatletter
\def\beginfloat{\centering\vspace{.1em}\makeparspace}
\def\figure@i[#1]{\oldfigure[#1]\beginfloat}
\def\figure@ii{\oldfigure\beginfloat}
\def\figure{\@ifnextchar[\figure@i \figure@ii}
\def\table@i[#1]{\oldtable[#1]\beginfloat}
\def\table@ii{\oldtable\beginfloat}
\def\table{\@ifnextchar[\table@i \table@ii}
\makeatother
\newcommand\beforefloat{\forcefontspace\makeparspace}
\newcommand\afterfloat{\vspace{.02em}}
\BeforeBeginEnvironment{figure}{\beforefloat}
\BeforeBeginEnvironment{table}{\beforefloat}
\AfterEndEnvironment{figure}{\afterfloat}
\AfterEndEnvironment{table}{\afterfloat}
\newlength{\parskipcopy}
\makeatletter
\def\@minipagerestore{\setlength{\intextsep}{0em}\setlength{\parskip}{\parskipcopy}\vphantom{ABCDEgjpqy}\vspace{-\baselineskip}\vspace{-\parskip}}
\AtEndEnvironment{minipage}{\forcefontspace}
\AfterEndEnvironment{minipage}{\makelinespace}
\let\oldfbox\fbox
\renewcommand{\fbox}[1]{\vspace{.05em}\setlength{\fboxsep}{0em}\oldfbox{#1}\vspace{.2em}\ifvmode{\makelinespace\ensurelinespace}\fi}
\renewcommand{\boxed}[1]{\oldfbox{\m@th$#1$}}	
\makeatother

\newtheoremstyle{lwq}
  {0em}	
  {0em}	
  {\normalfont}	
  {0em}	
  {\bfseries}	
  {.}	
  {.3em plus .2em}	
  {\thmname{#1}\thmnumber{ #2}\thmnote{ (#3)}}	
\newtheoremstyle{lwqprf}
  {0em}	
  {0em}	
  {\normalfont}	
  {0em}	
  {\itshape}	
  {.}	
  {.3em plus .2em}	
  {\thmname{#1}\thmnote{ (#3)}}	
\newlength{\thmspace}
\newlength{\prfspace}
\newcommand\thmbegin{\par\addvspace{\thmspace}\vspace{\parskip}\addpenalty{-200}}
\newcommand\prfbegin{\par\addvspace{\prfspace}\vspace{\parskip}}
\newcommand\thmend{\par\addvspace{\thmspace}\addpenalty{-100}}	
\newcommand\prfend{\par\addvspace{\prfspace}\addpenalty{-100}}	
\newcommand{\theoremname}{Theorem}
\theoremstyle{lwq}\newtheorem{thm}{\protect\theoremname}
\AtBeginEnvironment{thm}{\thmbegin} \AtEndEnvironment{thm}{\thmend}
\newcommand{\definitionname}{Definition}
\theoremstyle{lwq}\newtheorem{defn}[thm]{\protect\definitionname}
\AtBeginEnvironment{defn}{\thmbegin} \AtEndEnvironment{defn}{\thmend}
\newenvironment{centerbox}
{\par\begin{centering}}
{\par\end{centering}}
\newcommand{\lemmaname}{Lemma}
\theoremstyle{lwq}\newtheorem{lem}[thm]{\protect\lemmaname}
\AtBeginEnvironment{lem}{\thmbegin} \AtEndEnvironment{lem}{\thmend}
\renewcommand{\proofname}{Proof}
\theoremstyle{lwqprf}\newtheorem{prf}{\protect\proofname}
\renewenvironment{proof}[1][]{\prfbegin\begin{prf}[#1]\pushQED{\qed}}{\popQED\end{prf}\prfend}
\renewcommand{\qed}{\hfill{}\hspace{2em minus 1em}\qedsymbol}
\newcommand{\setblockspace}
{
  \setlength{\topsep}{0em}
  \setlength{\itemsep}{\listspace}
  \setlength{\parsep}{0em}
  \setlength{\partopsep}{\listspace}
}
\newcommand{\setblockflush}{}
\newlength{\blockmargin}
\newenvironment{block}
  {\list{}{\leftmargin\blockmargin\setblockspace\interlinepenalty3000}\setblockflush}
  {\endlist}
\newcommand{\remarkname}{Remark}
\theoremstyle{lwq}
\AtBeginEnvironment{rem}{\thmbegin} \AtEndEnvironment{rem}{\thmend}
\theoremstyle{lwq}\newtheorem*{rem*}{\protect\remarkname}
\AtBeginEnvironment{rem*}{\thmbegin} \AtEndEnvironment{rem*}{\thmend}
\usepackage[framemethod=TikZ]{mdframed}
\mdfdefinestyle{mdroundedboxinfloat}
{
  skipabove=0em,
  skipbelow=0em,
  roundcorner=.3em,
  innertopmargin=.3em,
  innerbottommargin=.3em,
  innerrightmargin=.3em,
  innerleftmargin=.3em,
  backgroundcolor=none,
  linecolor=gray,
}
\newenvironment{roundedboxinfloat}
{\vspace{-.3em}\begin{mdframed}[style=mdroundedboxinfloat]\vspace{.1em}\forcefontspace}
{\forcefontspace\end{mdframed}\unskip\vspace{-.24em}}

\AtBeginDocument{\fontsize{10}{12.6}\selectfont}
\AtBeginDocument{

}

\makeatother

\begin{document}
\renewcommand{\qed}{\hfill{}\hspace{2em minus 1em}\scalebox{1.4}{$\diamond$}}

\newcommand\br{\addpenalty{-1000}}

\newcommand\step{\crefname{enumi}{step}{}}

\newcommand\point{\crefname{enumi}{point}{}}

\global\long\def\nn{\mathbb{N}}
\global\long\def\zz{\mathbb{Z}}
\global\long\def\qq{\mathbb{Q}}
\global\long\def\rr{\mathbb{R}}

\global\long\def\wi{\subseteq}
\global\long\def\co{\supseteq}
\global\long\def\nwi{\nsubseteq}
\global\long\def\nco{\nsupseteq}
\global\long\def\none{\varnothing}
\global\long\def\less{\smallsetminus}

\global\long\def\ii{\mathbf{1}}
\global\long\def\pp{\mathbf{P}}
\global\long\def\ee{\mathbf{E}}
\global\long\def\vv{\mathbf{Var}}
\global\long\def\cv{\mathbf{Cov}}

\global\long\def\floor#1{\left\lfloor #1\right\rfloor }
\global\long\def\ceil#1{\left\lceil #1\right\rceil }

\global\long\def\cond#1#2#3{\left(\vphantom{#1#2#3}\right.\,#1\mathrel{\,?\,}\allowbreak#2\mathrel{\,:\,}\allowbreak#3\,\left.\vphantom{#1#2#3}\right)}

\global\long\def\f#1{\operatorname{#1}}

\global\long\def\a#1{.\text{#1}}

\global\long\def\e{\upvarepsilon}

\global\long\def\tbox#1{\boxed{\text{#1}}}

\global\long\def\stbox#1{\scalebox{0.8}{\boxed{\text{#1}}}}

\global\long\def\smbox#1{\scalebox{0.8}{\boxed{#1}}}

\global\long\def\grey#1{\textcolor{gray}{#1}}

\global\long\def\ubrace#1#2{\underbrace{\underset{}{{\strut}#1{\strut}}}_{#2}}

\global\long\def\vsp#1{\vspace{#1em}}

\global\long\def\tilt#1#2{\rotatebox{#2}{\ensuremath{#1}}}

\global\long\def\fs{\mathbb{FS}}

\noindent \begin{center}
\textbf{\huge{}Parallel Finger Search Structures}
\par\end{center}{\huge \par}

\noindent \begin{center}
\vsp{-.5}%
\begin{tabular}{>{\centering}p{0.3\textwidth}>{\centering}p{0.3\textwidth}}
\textbf{\large{}Seth Gilbert}{\large \par}

National University of Singapore & \textbf{\large{}Wei Quan Lim}{\large \par}

National University of Singapore\tabularnewline
\end{tabular}
\par\end{center}

\section*{Keywords}

Parallel data structures, multithreading, dictionaries, comparison-based
search, distribution-sensitive algorithms

\section*{Abstract}

In this paper~\footnote{This is the full version of a paper published in the 33rd International
Symposium on Distributed Computing (DISC 2019). It is posted here
 for your personal or classroom use. Not for redistribution.  \copyright~2019
Copyright is held by the owner/author(s). } we present two versions of a \textbf{parallel finger structure} $\fs$
on $p$ processors that supports searches, insertions and deletions,
and has a finger at each end. This is to our knowledge the first implementation
of a parallel search structure that is \textbf{\textit{work-optimal}}
with respect to the finger bound and yet has very good parallelism
(within a factor of $O\left((\log p)^{2}\right)$ of optimal). We
utilize an \textbf{extended implicit batching} framework that transparently
facilitates the use of $\fs$ by any parallel program $P$ that is
modelled by a dynamically generated DAG $D$ where each node is either
a unit-time instruction or a call to $\fs$.

The total work done by either version of $\fs$ is bounded by the
\textbf{finger bound} $F_{L}$ (for some linearization $L$ of $D$),
i.e.~each operation on an item with distance $r$ from a finger takes
$O(\log r+1)$ amortized work. Running $P$ using the simpler version
takes $O\left(\frac{T_{1}+F_{L}}{p}+T_{\infty}+d\cdot\left((\log p)^{2}+\log n\right)\right)$
time on a greedy scheduler, where $T_{1},T_{\infty}$ are the size
and span of $D$ respectively, and $n$ is the maximum number of items
in $\fs$, and $d$ is the maximum number of calls to $\fs$ along
any path in $D$. Using the faster version, this is reduced to $O\left(\frac{T_{1}+F_{L}}{p}+T_{\infty}+d\cdot(\log p)^{2}+s_{L}\right)$
time, where $s_{L}$ is the weighted span of $D$ where each call
to $\fs$ is weighted by its cost according to $F_{L}$. We also sketch
how to extend $\fs$ to support a fixed number of movable fingers.

The data structures in our paper fit into the \textbf{\textit{dynamic
multithreading}} paradigm, and their performance bounds are directly
\textbf{\textit{composable}} with other data structures given in the
same paradigm. Also, the results can be translated to practical implementations
using work-stealing schedulers.

\section*{Acknowledgements}

We would like to express our gratitude to our families and friends
for their wholehearted support, to the kind reviewers who provided
helpful feedback, and to all others who have given us valuable comments
and advice. This research was supported in part by Singapore MOE AcRF
Tier 1 grant T1 251RES1719.

\section{Introduction}

There has been much research on designing parallel programs and parallel
data structures. The \textbf{dynamic multithreading paradigm} (see~\cite{CormenLeRi09}
chap.~27) is one common parallel programming model, in which algorithmic
parallelism is expressed through parallel programming primitives such
as fork/join (also spawn/sync), parallel loops and synchronized methods,
but the program cannot stipulate any mapping from subcomputations
to processors. This is the case with many parallel languages and libraries,
such as Cilk dialects~\cite{Cilk,IntelCilkPlus13}, Intel TBB~\cite{TBB},
Microsoft Task Parallel Library~\cite{TPL} and subsets of OpenMP~\cite{OpenMP}.

Recently, Agrawal et al.~\cite{agrawal2014batcher} introduced the
exciting \textbf{\textit{modular design}} approach of \textbf{\textit{\emph{implicit
batching}}}, in which the programmer writes a multithreaded parallel
program that uses a \textbf{\textit{black box}} data structure, treating
calls to the data structure as basic operations, and also provides
a data structure that supports batched operations. Given these, the
runtime system automatically combines these two components together,
buffering data structure operations generated by the program, and
executing them in batches on the data structure.

This idea was extended in \cite{OPWM} to data structures that do
not process only one batch at a time (to improve parallelism). In
this \textbf{extended implicit batching framework}, the runtime system
not only holds the data structure operations in a \textbf{parallel
buffer}, to form the next \textbf{input batch}, but also \textbf{notifies}
the data structure on receiving the first operation in each batch.
Independently, the data structure can at any point \textbf{flush}
the parallel buffer to get the next batch.

This framework nicely supports \textbf{\textit{pipelined}} batched
data structures, since the data structure can decide when it is ready
to get the next input batch from the parallel buffer, which may be
even before it has finished processing the previous batch. Furthermore,
this framework makes it easy for us to build \textbf{\textit{composable}}
parallel algorithms and data structures with composable performance
bounds. This is demonstrated by both the parallel working-set map
in \cite{OPWM} and the parallel finger structure in this paper.

\subsection*{Finger Structures}

The \textbf{map} (or \textbf{dictionary}) data structure, which supports
inserts, deletes and searches/updates, collectively referred to as
\textbf{accesses}, comes in many different kinds. A common implementation
of a map is a balanced binary search tree such as an AVL tree or a
red-black tree, which (in the comparison model) takes $O(\log n)$
worst-case cost per access for a tree with $n$ items. There are also
maps such as splay trees~\cite{sleator1985splay} that have amortized
rather than worst-case performance bounds.

A \textbf{finger structure} is a special kind of map that comes with
a \textbf{fixed finger} at each end and a (fixed) number of \textbf{movable
fingers}, each of which has a key (possibly $-\infty$ or $\infty$
or between adjacent items in the map) that determines its position
in the map, such that accessing items nearer the fingers is cheaper.
For instance, the finger tree~\cite{kosaraju1981finger} was designed
to have the finger property in the worst case; it takes $O(\log r+1)$
steps per operation with finger distance $r$ (\ref{def:finger-dist}),
so its total cost satisfies the finger bound (\ref{def:finger-bound}).

\begin{defn}[Finger Distance]
\label{def:finger-dist} Define the \textbf{finger distance of accessing
an item} $x$ on a finger structure $M$ to be the number of items
from $x$ to the nearest finger in $M$ (including $x$), and the
\textbf{finger distance of moving a finger} to be the distance moved.
\end{defn}

\begin{defn}[Finger Bound]
\label{def:finger-bound} Given any sequence $L$ of $N$ operations
on a finger structure $M$, let $F_{L}$ denote the \textbf{finger
bound} for $L$, defined by $F_{L}=\sum_{i=1}^{N}\left(\log r_{i}+1\right)$
where $r_{i}$ is the finger distance of the $i$-th operation in
$L$ when $L$ is performed on $M$.
\end{defn}

\subsection*{Main Results \phantomsection}

\label{sec:main}

We present in this paper, to the best of our knowledge, the first
parallel finger structure. In particular, we design two parallel maps
that are \textbf{\textit{work-optimal}} with respect to the \nameref{def:finger-bound}
$F_{L}$ (i.e.~it takes $O(F_{L})$ work) for some linearization
$L$ of the operations (that is consistent with the results), while
having very good parallelism. (We assume that each key comparison
takes $O(1)$ steps.)

These parallel finger structures can be used by any parallel program
$P$, whose actual execution is captured by a \textbf{program DAG}
$D$, where each node is an instruction that finishes in $O(1)$ time
or a call to the finger structure $M$, called an \textbf{$M$-call},
that blocks until the result is returned, and each edge represents
a dependency due to the parallel programming primitives.

The first design, called $\fs_{1}$, is a simpler data structure that
processes operations one batch at a time.
\begin{thm}[$\fs_{1}$ Performance]
\label{thm:FS1-perf} If $P$ uses $\fs_{1}$ (as $M$), then its
running time on $p$ processes using any greedy scheduler (i.e.~at
each step, as many tasks are executed as are available, up to $p$)
is 
\[
O\left(\frac{T_{1}+F_{L}}{p}+T_{\infty}+d\cdot\left((\log p)^{2}+\log n\right)\right)
\]
for some linearization $L$ of $M$-calls in $D$, where $T_{1}$
is the number of nodes in $D$, and $T_{\infty}$ is the number of
nodes on the longest path in $D$, and $d$ is the maximum number
of $M$-calls on any path in $D$, and $n$ is the maximum size of
$M$.~\footnote{To cater to instructions that may not finish in $O(1)$ time (e.g.
due to memory contention), it suffices to define $T_{1}$ and $T_{\infty}$
to be the (weighted) work and span (\ref{def:work-span}) respectively
of the program DAG where each $M$-call is assumed to take $O(1)$
time.}
\end{thm}
Notice that if $M$ is an ideal concurrent finger structure (i.e.~one
that takes $O\left(F_{L}\right)$ work), then running $P$ using $M$
on $p$ processors according to the linearization $L$ takes $\Omega(T_{opt})$
worst-case time where $T_{opt}=\frac{T_{1}+F_{L}}{p}+T_{\infty}$.
Thus $\fs_{1}$ gives an essentially optimal time bound except for
the `span term' $d\cdot\left((\log p)^{2}+\log n\right)$, which
adds $O\left((\log p)^{2}+\log n\right)$ time per $\fs_{1}$-call
along some path in $D$.

The second design, called $\fs_{2}$, uses a complex internal pipeline
to reduce the `span term'.
\begin{thm}[$\fs_{2}$ Performance]
\label{thm:FS2-perf} If $P$ uses $\fs_{2}$, then its running time
on $p$ processes using any greedy scheduler is 
\[
O\left(\frac{T_{1}+F_{L}}{p}+T_{\infty}+d\cdot(\log p)^{2}+s_{L}\right)
\]
for some linearization $L$ of $M$-calls in $D$, where $d$ is the
maximum number of $\fs_{2}$-calls on any path in $D$, and $s_{L}$
is the weighted span of $D$ where each $\fs_{2}$-call is weighted
by its cost according to $F_{L}$, except that each finger-move operation
is weighted by $\log n$. Specifically, each access $\fs_{2}$-call
that is an access with finger distance $r$ according to $L$ is given
the weight $\log r+1$, and each $\fs_{2}$-call that is a finger-move
is given the weight $\log n$, and $s_{L}$ is the maximum weight
of any path in $D$. Thus, ignoring finger-move operations, $\fs_{2}$
gives an essentially optimal time bound up to an extra $O\left((\log p)^{2}\right)$
time per $\fs_{2}$-call along some path in $D$.
\end{thm}
We shall first focus on basic finger structures with just one fixed
finger at each end, since we can implement the general finger structure
with $f$ movable fingers by essentially concatenating $(f+1)$ basic
finger structures, as we shall explain later in \ref{sec:GPFS}. We
will also discuss later in \ref{sec:work-steal} how to adapt our
results for work-stealing schedulers that can actually be provided
by a real runtime system.

\subsection*{Challenges \& Key Ideas}

The sequential finger structure in \cite{guibas1977finger} (essentially
a B-tree with carefully staggered rebalancing) takes $O(\log r+1)$
worst-case time per access with finger distance $r$, but seems impossible
to parallelize efficiently. It turns out that relaxing this bound
to $O(\log r+1)$ amortized time admits a \textbf{\textit{simple sequential
finger structure}} $\fs_{0}$ (\ref{sec:ASFS}) that can be parallelized.
In $\fs_{0}$, the items are stored in order in a list of segments
$S_{0}[0],S_{0}[1],\cdots,S_{0}[l]\,,\,S_{1}[l],\cdots,S_{1}[1],S_{1}[0]$,
where each segment $S_{i}[k]$ is a balanced binary search tree with
size at most $3\cdot c(k)$ but at least $c(k)$ unless $k=l$, where
$c(k)=2^{2^{k+1}}$. This ensures that $S_{i}[k]$ has height $O\left(2^{k}\right)$,
and that the $r$ least items are in the first $\log O(\log r)$ segments
and the $r$ greatest items are in the last $\log O(\log r)$ segments.
Thus for each operation with finger distance $r$, it takes $O(\log r+1)$
time to search through the segments from both ends simultaneously
to find the correct segment and perform the operation in it. After
that, we rebalance the segments to preserve the size invariant, in
such a way that each imbalanced segment $S_{i}[k]$ will have new
size $2\cdot c(k)$. This \textbf{\textit{double-exponential}} segment
sizes and the \textbf{\textit{reset-to-middle}} rebalancing is critical
in ensuring that all the rebalancing takes $O(1)$ amortized time
per operation, even if each rebalancing cascade may take up to $\Theta(\log n)$
time.

\textbf{\textit{The challenge is to parallelize $\fs_{0}$ while preserving
the total work}}. Naturally, we want to process operations in batches,
and use a batch-parallel search structure in place of each binary
search tree. This may seem superficially similar to the parallel working-set
map in \cite{OPWM}, but the techniques in the earlier paper cannot
be applied in the same way, for three main reasons.

Firstly, searches and deletions for items not in the map must still
be cheap if they have small finger distance, so we have to \textbf{\textit{eliminate
these operation in a separate preliminary phase}} by an unsorted search
of the smaller segments, before sorting and executing the other operations.

Secondly, insertions and deletions must be cheap if they have small
finger distance (e.g. deleting an item from the first segment must
have $O(1)$ cost), so we \textbf{\textit{cannot enforce a tight segment
size invariant}}, otherwise rebalancing would be too costly.

This is unlike the parallel working-set map, where we not only have
a budget of $O(\log n)$ for each insertion or deletion or failed
search, but also must shift accessed items sufficiently near to the
front to achieve the desired span bound. The rebalancing in the parallel
finger structures in this paper is hence completely different from
that in the parallel working-set map.

Thirdly, for the faster version $\fs_{2}$ where the larger segments
are pipelined, in order to keep all segments sufficiently balanced,
the pipelined segments must never be too underfull, so we must \textbf{\textit{carefully
restrict when a batch is allowed to be processed at a segment}}. Due
to this, we cannot even guarantee that a batch of operations will
proceed at a consistent pace through the pipeline, but we can use
an accounting argument to bound the `excess delay' by the number
of $\fs_{2}$-calls divided by $p$.

\subsection*{Other Related Work}

There are many approaches for designing efficient parallel data structures,
so as to make maximal use of parallelism in a multi-processor system,
whether with empirical or theoretical efficiency.

For example, Ellen et al.~\cite{EllenPR10} show how to design a
non-blocking concurrent binary search tree, with later work analyzing
the amortized complexity~\cite{ellen2014nonblockingbst} and generalizing
this technique~\cite{brown2014nonblockingtrees}. Another notable
concurrent search tree is the CBTree~\cite{afek2012cbtree,afek2014cbtree},
which is based on the splay tree. But despite experimental success,
the theoretical access cost for these tree structures may increase
with the number of concurrent operations due to contention near the
root, and some of them do not even maintain balance (i.e., the height
may get large).

Another method is software combining~\cite{FatourouKa12,HendlerInSh10,OyamaTaYo99},
where each process inserts a request into a shared queue and at any
time one process is sequentially executing the outstanding requests.
This generalizes to parallel combining~\cite{aksenov2018parallelcombining},
where outstanding requests are executed in batches on a suitable batch-parallel
data structure (similar to implicit batching). These methods were
shown to yield empirically efficient concurrent implementations of
various common abstract data structures including stacks, queues and
priority queues.

In the PRAM model, Paul et al.~\cite{paul1983paralleldict} devised
a parallel 2-3 tree where $p$ synchronous processors can perform
a sorted batch of $p$ operations on a parallel 2-3 tree of size $n$
in $O(\log n+\log p)$ time. Blelloch et al.~\cite{BlellochRe97}
show how to increase parallelism of tree operations via pipelining.
Other similar data structures include parallel treaps~\cite{BlellochR98}
and a variety of work-optimal parallel ordered sets~\cite{blelloch2016justjoin}
supporting unions and intersections with optimal work, but these do
not have optimal span. As it turns out, we can in fact have parallel
ordered sets with optimal work and span~\cite{akhremtsev2016fastparallelsets,P23T}.

Nevertheless, the programmer cannot use this kind of parallel data
structure as a black box with atomic operations in a high-level parallel
program, but must instead carefully coordinate access to it. This
difficulty can be eliminated by designing a suitable \textbf{\textit{batch-parallel}}
data structure and using \textbf{\textit{implicit batching}}~\cite{agrawal2014batcher}
or \textbf{\textit{extended implicit batching}} as presented in \cite{OPWM}
and more fully in this paper. Batch-parallel implementations have
been designed for various data structures including weight-balanced
B-trees~\cite{erb2014parallelwbtrees}, priority queues~\cite{aksenov2018parallelcombining},
working-set maps~\cite{OPWM} and euler-tour trees~\cite{tseng2019paralleleulertree}.

\section{Parallel Computation Model}

\label{sec:model}

In this section, we describe parallel programming primitives in our
model, how a parallel program generates an execution DAG, and how
we measure the cost of an execution DAG.

\subsection{Parallel Primitives}

\label{sub:primitives}

The parallel finger structures $\fs_{1}$ and $\fs_{2}$ in this paper
are described and explained as multithreaded data structures that
can be used as composable building blocks in a larger parallel program.
In this paper we shall focus on the abstract algorithms behind $\fs_{1}$
and $\fs_{2}$, relying merely on the following parallel programming
primitives (rather than model-specific implementation details, but
see Appendix \ref{sub:locking} for those):
\begin{enumerate}
\item \textbf{Threads}: A thread can at any point \textbf{\textit{terminate}}
itself (i.e.~finish running). Or it can \textbf{\textit{fork}} a
new thread, obtaining a pointer to that thread, or \textbf{\textit{join}}
to another thread (i.e.~wait until that thread terminates). Or it
can \textbf{\textit{suspend}} itself (i.e.~temporarily stop running),
after which a thread with a pointer to it can \textbf{\textit{resume}}
it (i.e.~make it continue running from where it left off). Each of
these takes $O(1)$ time.
\item \textbf{Non-blocking locks}: Attempts to \textbf{\textit{acquire}}
a non-blocking lock are serialized but do not block. Acquiring the
lock succeeds if the lock is not currently held but fails otherwise,
and \textbf{\textit{releasing}} always succeeds. If $k$ threads concurrently
access the lock, then each access finishes within $O(k)$ time.
\item \textbf{Dedicated lock}: A dedicated lock is a blocking lock initialized
with a constant number of keys, where concurrent threads must use
different keys to \textbf{\textit{acquire}} it, but \textbf{\textit{releasing}}
does not require a key. Each attempt to acquire the lock takes $O(1)$
time, and the thread will acquire the lock after at most $O(1)$ subsequent
acquisitions of that lock.
\item \textbf{Reactivation calls}: A procedure $P$ with no input/output
can be encapsulated by a reactivation wrapper, in which it can be
run only via \textbf{\textit{reactivations}}. If there are always
at most $O(1)$ concurrent reactivations of $P$, then whenever a
thread \textbf{\textit{reactivates}} $P$, if $P$ is not currently
running then it will start running (in another thread forked in $O(1)$
time), otherwise it will run within $O(1)$ time after its current
run finishes.
\end{enumerate}
We also make use of basic batch operations, namely filtering, sorted
partitioning, joining and merging (see Appendix \ref{sub:par-batch}),
which have easy implementations using arrays in the binary forking
model in \cite{blelloch2019mtramalgo}. So $\fs_{1}$ and $\fs_{2}$
(using a work-stealing scheduler) can be implemented in the Arbitrary
CRCW PRAM model with fetch-and-add, achieving the claimed performance
bounds. Actually, $\fs_{1}$ and $\fs_{2}$ were also designed to
function correctly with the same performance bounds in a much stricter
computation model called the QRMW parallel pointer machine model (see
Appendix \ref{sub:qrmw-ppm} for details).

\subsection{Execution DAG}

The \textbf{program DAG} $D$ captures the high-level execution of
$P$, but the actual complete execution of $P$ (including interaction
between data structure calls) is captured by the \textbf{execution
DAG} $E$ (which may be schedule-dependent), in which each node is
a basic instruction and the directed edges represent the computation
dependencies (such as constrained by forking/joining of threads and
acquiring/releasing of blocking locks). At any point during the execution
of $P$, a node in the program/execution DAG is said to be \textbf{ready}
if its parent nodes have been executed. At any point in the execution,
an \textbf{active thread} is simply a ready node in $E$, while a
\textbf{terminated/suspended thread} is an executed node in $E$ that
has no child nodes.

The execution DAG $E$ consists of \textbf{program nodes} (specifically
\textbf{$P$-nodes}) and \textbf{ds (data-structure) nodes}, which
are dynamically generated as follows. At the start $E$ has a single
program node, corresponding to the start of the program $P$. Each
node could be a \textbf{normal instruction} (i.e.~basic arithmetic/memory
operation) or a \textbf{parallel primitive} (see \ref{sub:primitives}).
Each program node could also be a \textbf{data structure call}.

When a (ready) node is executed, it may generate child nodes or \textbf{\textit{terminate}}.
A normal instruction generates one child node and no extra edges.
A \textbf{\textit{join}} generates a child node with an extra edge
to it from the \textbf{\textit{terminate}} node of the joined thread.
A \textbf{\textit{resume}} generates an extra child node (the resumed
thread) with an edge to it from the \textbf{\textit{suspend}} node
of the originally suspended thread. Accesses to locks and reactivation
calls would each expand to a subDAG comprised of normal instructions
and possibly \textbf{\textit{fork}}/\textbf{\textit{suspend}}/\textbf{\textit{resume}}.

The program nodes correspond to nodes in the program DAG $D$, and
except for data structure calls they generate only program nodes.
A call to a data structure $M$ is called an \textbf{$M$-call}. If
$M$ is an ordinary (non-batched) data structure, then an $M$-call
generates an \textbf{$M$-node} (and every $M$-node is a ds node),
which thereafter generates only $M$-nodes except for calls to other
data structures (external to $M$) or returning the result of some
operation (generating a program node with an edge to it from the original
$M$-call).

However, if $M$ is an \textbf{\textit{(implicitly) batched}} data
structure, then all $M$-calls are automatically passed to the \textbf{parallel
buffer} for $M$ (see Appendix \ref{sub:par-buffer}). So an $M$-call
generates a \textbf{buffer node} corresponding to passing the call
to the parallel buffer, as if the parallel buffer for $M$ is itself
another data structure and not part of $M$. Buffer nodes generate
only buffer nodes until it notifies $M$ of the buffered $M$-calls
or passes the input batch to $M$, which generates an $M$-node. In
short, $M$-nodes exclude all nodes generated as part of the buffer
subcomputations (i.e.~buffering the $M$-calls, and notifying $M$,
and flushing the buffer).

\subsection{Data Structure Costs}

We shall now define work and span of any (terminating) subcomputation
of a multithreaded program, i.e.~any subset of the nodes in its execution
DAG. This allows us to capture the intrinsic costs incurred by a data
structure, separate from the costs of a parallel program using it.
\begin{defn}[Subcomputation Work/Span/Cost]
\label{def:work-span} Take any execution of a parallel program $P$
(on $p$ processors), and take any subset $C$ of nodes in its execution
DAG $E$. The \textbf{work} taken by $C$ is the total weight $w$
of $C$ where each node is weighted by the time taken to execute it.
The \textbf{span} taken by $C$ is the maximum weight $s$ of nodes
in $C$ on any (directed) path in $E$. The \textbf{cost} of $C$
is $\frac{w}{p}+s$.
\end{defn}

\begin{defn}[Data Structure Work/Span/Cost]
\label{def:effective} Take any parallel program $P$ using a data
structure $M$. The \textbf{work/span/cost} of $M$ (as used by $P$)
is the work/span/cost of the $M$-nodes in the execution DAG for $P$.
\end{defn}
Note that the cost of the entire execution DAG is in fact an upper
bound on the actual time taken to run it on a \textbf{greedy scheduler},
which on each step assigns as many unassigned ready nodes (i.e.~nodes
that have been generated but have not been assigned) as possible to
available processors (i.e.~processors that are not executing any
nodes) to be executed.

Moreover, the subcomputation cost is \textbf{\textit{subadditive}}
across subcomputations. Thus our results are \textbf{\textit{composable}}
with other algorithms and data structures in this model, since we
actually show the following for some linearization $L$ (where $F_{L},d,n,s_{L}$
are as defined in \ref{sec:main} Main Results, and $N$ is the total
number of calls to the parallel finger structure).
\begin{thm}[$\fs$ Work/Span Bounds]
\label{thm:FS-costs}~
\begin{itemize}
\item (\ref{thm:FS1-work} and \ref{thm:FS1-span}) $\fs_{1}$ takes $O\left(F_{L}\right)$
work and $O\left(\frac{N}{p}+d\cdot\left((\log p)^{2}+\log n\right)\right)$
span.
\item (\ref{thm:FS2-work} and \ref{thm:FS2-span}) $\fs_{2}$ takes $O\left(F_{L}\right)$
work and $O\left(\frac{N}{p}+d\cdot(\log p)^{2}+s_{L}\right)$ span.
\end{itemize}
\end{thm}
Note that the bounds for the work/span of $\fs_{1}$ and $\fs_{2}$
are independent of the scheduler. In addition, using any greedy scheduler,
the parallel buffer for either finger structure has cost $O\left(\frac{T_{1}+F_{L}}{p}+d\cdot\log p\right)$
(Appendix \ref{rem:par-buff-cost}). Therefore our main results (\ref{thm:FS1-perf}
and \ref{thm:FS2-perf}) follow from these composable bounds (\ref{thm:FS-costs}).

In general, if a program uses a fixed number of implicitly batched
data structures, then running it using a greedy scheduler takes $O\left(\frac{T_{1}+w^{*}}{p}+T_{\infty}+s^{*}+d^{*}\cdot\log p\right)$
time, where $w^{*}$ is the total work of all the data structures,
and $s^{*}$ is the total span of all the data structures, and $d^{*}$
is the maximum number of data structure calls on any path in the program
DAG.

\section{Amortized Sequential Finger Structure}

\label{sec:ASFS}

In this section we explain a sequential finger structure $\fs_{0}$
with a fixed finger at each end, which (unlike finger structures based
on balanced binary trees) is amenable to parallelization and pipelining
due to its \textbf{\textit{doubly-exponential segmented structure}}
(which was partially inspired by Iacono's working-set structure~\cite{iacono2001wstree}).

\begin{figure}[H]
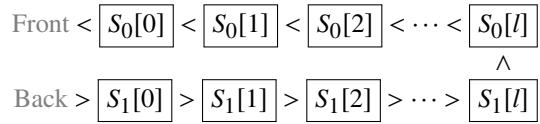

\begin{centerbox}
\begin{roundedboxinfloat}
\begin{centerbox}
$\begin{matrix}\grey{\text{Front}}<\boxed{S_{0}[0]}<\boxed{S_{0}[1]}<\boxed{S_{0}[2]}<\cdots<{}\\
\\
\grey{\text{Back}}>\boxed{S_{1}[0]}>\boxed{S_{1}[1]}>\boxed{S_{1}[2]}>\cdots>{}
\end{matrix}\begin{matrix}\boxed{S_{0}[l]}\\
\wedge\\
\boxed{S_{1}[l]}
\end{matrix}$
\end{centerbox}
\caption{\label{fig:FS0-outline}$\protect\fs_{0}$ Outline; each box $S_{i}[k]$
represents a 2-3 tree of size $\Theta(2^{2^{k}})$ for $k<l$}
\end{roundedboxinfloat}
\end{centerbox}
\end{figure}

$\fs_{0}$ keeps the items in order in two halves, the front half
stored in a chain of \textbf{segments} $S_{0}[0..l]$, and the back
half stored in reverse order in a chain of segments $S_{1}[0..l]$.
Let $c(k)=2^{2^{k+1}}$ for each $k\in\zz$. Each segment $S_{i}[k]$
has a \textbf{target size} $t(k)=2\cdot c(k)$, and a \textbf{target
capacity} defined to be $[t(k),t(k)]$ if $k<l$ but $[0,t(k)]$ if
$k=l$. Each segment stores its items in order in a 2-3 tree. We say
that a segment $S_{i}[k]$ is \textbf{balanced} iff its size is within
$c(k)$ of its target capacity, and\textbf{ overfull} iff it has more
than $c(k)$ items above target capacity, and \textbf{underfull} iff
it has more than $c(k)$ items below target capacity. At any time
we associate every item $x$ to a unique segment that it \textbf{fits}
in; $x$ fits in $S_{0}[k]$ if $k$ is the minimum such that $x\le\max(S_{0}[k])$,
and that $x$ fits in $S_{1}[k]$ if $k$ is the minimum such that
$x\ge\min(S_{1}[k])$, and that $x$ fits in $S_{0}[l]$ if $\max(S_{0}[l])<x<\min(S_{1}[l])$.
We shall maintain the invariant that every segment is balanced after
each operation is finished.

For each operation on an item $x$, we find the segment $S_{i}[k]$
that $x$ fits in, by checking the range of items in $S_{0}[a]$ and
$S_{1}[a]$ for each $a$ from $0$ to $l$ and stopping once $k$
is found, and then perform the desired operation on the 2-3 tree in
$S_{i}[k]$. This takes $O(k+\log(t(k)+c(k)))\wi O\left(2^{k}\right)$
steps, and $2^{k}=\log_{2}c(k-1)\le\log_{2}r+1$ where $r$ is the
finger distance of the operation.

After that, if $S_{i}[k]$ becomes imbalanced, we \textbf{rebalance}
it by shifting (appropriate) items to or from $S_{i}[k+1]$ (after
creating empty segment $S_{i}[k+1]$ if it does not exist) to make
$S_{i}[k]$ have target size or as close as possible (via a suitable
split then join of the 2-3 trees), and then $S_{i}[k+1]$ is removed
if it is the last segment and is now empty. After the rebalancing,
$S_{i}[k]$ will not only be balanced but also have size within its
target capacity. But now $S_{i}[k+1]$ may become imbalanced, so the
rebalancing may cascade.

Finally, if one chain $S_{i}[0..l']$ is longer than the other chain
$S_{j}[0..l]$, it must be that $l'=l+1$, so we \textbf{rebalance}
the chains as follows: If $S_{j}[l]$ is below target size, shift
items from $S_{i}[l']$ to $S_{j}[l]$ to fill it up to target size.
If $S_{j}[l]$ is (still) below target size, remove the now empty
$S_{i}[l']$, otherwise add a new empty segment $S_{j}[l+1]$.

Rebalancing may cascade throughout the whole chain and take $\Theta(\log n)$
steps. But we shall show below that the rebalancing costs can be amortized
away completely, and hence each operation with finger distance $r$
takes $O(\log r+1)$ amortized steps, giving us the finger bound for
$\fs_{0}$. We will later use the same technique in analyzing $\fs_{1}$
and $\fs_{2}$ as well.
\begin{lem}[$\fs_{0}$ Rebalancing Cost]
 All the rebalancing takes $O(1)$ amortized steps per operation.\end{lem}
\begin{proof}
We shall maintain the invariant that each segment $S_{i}[k]$ with
$q$ items beyond (i.e.~above or below) its target capacity has at
least $q\cdot2^{-k}$ stored credits. Each operation is given $1$
credit, and we use it to pay for any needed extra stored credits at
the segment where we perform the operation. Whenever a segment $S_{i}[k]$
is rebalanced, it must have had $q$ items beyond its target capacity
for some $q>c(k)$, and so had at least $q\cdot2^{-k}$ stored credits.
Also, the rebalancing itself takes $O(\log(t(k)+q)+\log(t(k+1)+c(k+1)+q))\wi O(\log q)\wi O\left(q\cdot2^{-k}\right)$
steps, after which $S_{i}[k+1]$ needs at most $q\cdot2^{-(k+1)}$
extra stored credits. Thus the stored credits at $S_{i}[k]$ can be
used to pay for both the rebalancing and any extra stored credits
needed by $S_{i}[k+1]$. Whenever the chains are rebalanced, it can
be paid for by the last segment rebalancing (which created or removed
a segment), and no extra stored credits are needed. Therefore the
total rebalancing cost amounts to $O(1)$ per operation.
\end{proof}

\section{Simpler Parallel Finger Structure}

\label{sec:PFS1}

We now present our simpler parallel finger structure $\fs_{1}$. The
idea is to use the amortized sequential finger structure $\fs_{0}$
(\ref{sec:ASFS}) and execute operations in batches. We group each
pair of segments $S_{0}[k]$ and $S_{1}[k]$ into one \textbf{section}
$S[k]$, and we say that an item $x$ \textbf{fits in} the sections
$S[j..k]$ iff $x$ fits in some segment in $S[j..k]$.

The items in each segment are stored in a \textbf{batch-parallel map}
(Appendix \ref{sub:batch-map}), which supports:
\begin{itemize}
\item \textbf{Unsorted batch search:} Search for an unsorted batch of $b$
items within $O(b\cdot\log n)$ work and $O(\log b\cdot\log n)$ span,
tagging each search with the result, where $n$ is the map size.
\item \textbf{Sorted batch access:} Perform an item-sorted batch of $b$
operations on distinct items within $O\left(b\cdot\log n\right)$
work and $O(\log b+\log n)$ span, tagging each operation with the
result, where $n$ is the map size.
\item \textbf{Split}: Split a map of size $n$ around a given pivot rank
(into lower+upper parts) within $O(\log n)$ work/span.
\item \textbf{Join}: Join maps of total size $n$ separated by a pivot (i.e.~lower+upper
parts) within $O(\log n)$ work/span.
\end{itemize}
For each section $S[k]$, we can perform a batch of $b$ operations
on it within $O(b\cdot\log c(k))$ work and $O(\log b+\log c(k))$
span if we have the batch sorted. Excluding sorting, the total work
would satisfy the finger bound for the same reason as in $\fs_{0}$.
However, we cannot afford to sort the input batch right at the start,
because if the batch had $b$ searches of distinct items all with
finger distance $O(1)$, then it would take $\Omega(b\cdot\log b)$
work and exceed our finger bound budget of $O(b)$.

We can solve this by splitting the sections into two slabs, where
the first slab comprises the first $\log\log(2b)$ sections, and passing
the batch through a preliminary phase in which we merely perform an
unsorted search of the relevant items in the first slab, and eliminate
operations on items that fit in the first slab but are neither found
nor to be inserted.

This preliminary phase takes $O(\log c(k))$ work per operation and
$O(\log b\cdot\log c(k))$ span at each section $S[k]$. We then sort
the uneliminated operations and execute them on the appropriate slab.
For this, ordinary sorting still takes too much work as there can
be many operations on the same item, but it turns out that the finger
bound budget is enough to pay for entropy-sorting (Appendix \ref{def:par-esort}),
which takes $O\left(\log\frac{b}{q}+1\right)$ work for each item
that occurs $q$ times in the batch. Rebalancing the segments and
chains is a little tricky, but if done correctly it takes $O(1)$
amortized work per operation. Therefore we achieve work-optimality
while being able to process each batch within $O\left((\log b)^{2}+\log n\right)$
span. The details are below.

\subsection{Description of $\protect\fs_{1}$}

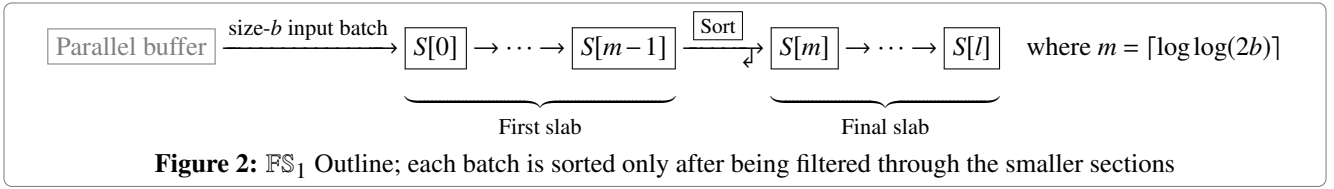
\label{sub:FS1-desc}

\begin{figure}[H]
\begin{centerbox}
\begin{roundedboxinfloat}
\begin{centerbox}
$\grey{\tbox{Parallel\ buffer}}\xrightarrow{\text{size-\ensuremath{b} input batch}}\ubrace{\boxed{S[0]}\to\cdots\to\boxed{S[m-1]}}{\text{First slab}}\xrightarrow[\smash{\qquad\raisebox{.84em}{\tilt{\Rsh}{180}}}]{\stbox{Sort}\ }\ubrace{\boxed{S[m]}\to\cdots\to\boxed{S[l]}}{\text{Final slab}}$\quad{}where
$m=\ceil{\log\log(2b)}$
\end{centerbox}
\caption{\label{fig:FS1-flow}$\protect\fs_{1}$ Outline; each batch is sorted
only after being filtered through the smaller sections}
\end{roundedboxinfloat}
\end{centerbox}
\end{figure}

$\fs_{1}$-calls are put into the parallel buffer (\ref{sec:model})
for $\fs_{1}$. Whenever the previous batch is done, $\fs_{1}$ flushes
the parallel buffer to obtain the next batch $B$. Let $b$ be the
size of $B$, and we can assume $b>1$. Based on $b$, the sections
in $\fs_{1}$ are conceptually divided into two slabs, the \textbf{first
slab} comprising sections $S[0..m-1]$ and the \textbf{final slab}
comprising sections $S[m..l]$, where $m=\ceil{\log\log(2b)}+1$ (where
$\log$ is the binary logarithm). The items in each segment are stored
in a batch-parallel map (Appendix \ref{sub:batch-map}).

$\fs_{1}$ processes the input batch $B$ in four phases: \phantomsection\label{par:FS1-phases}
\begin{enumerate}
\item \textbf{Preliminary phase}: For each first slab section $S[k]$ in
order (i.e.~$k$ from $0$ to $m-1$) do as follows:

\begin{enumerate}
\item Perform an unsorted search in each segment in $S[k]$ for all the
items relevant to the remaining batch $B'$ (of direct pointers into
$B$), and tag the operations in the original batch $B$ with the
results.
\item Remove all operations on items that fit in $S[k]$ from the remaining
batch $B'$.
\item Skip the rest of the first slab if $B'$ becomes empty.
\end{enumerate}
\item \textbf{Separation phase}: Partition $B$ based on the tags into three
parts and handle each part separately as follows:

\begin{enumerate}
\item \textbf{\textit{Ineffectual operations}} (on items that fit in the
first slab but are neither found nor to be inserted): Return the results.

\item \textbf{\textit{Effectual operations}} (on items found in or to be
inserted into the first slab): Entropy-sort (Appendix \ref{def:par-esort})
them in order of access type (search, update, insertion, deletion)
with deletions last, followed by item, combining operations of the
same access type on the same item into one \textbf{group-operation}
that is treated as a single operation whose \textbf{effect} is the
last operation in that group. Each group-operation is stored in a
leaf-based binary tree with height $O(\log b)$ (but not necessarily
balanced), and the combining is done during the entropy-sorting itself.

\item \textbf{\textit{Residual operations}} (on items that do not fit in
the first slab): Sort them while combining operations in the same
manner as for effectual operations.~\footnote{This does not require entropy-sorting, but combining merge-sort essentially
achieves the entropy bound anyway.}
\end{enumerate}
\item \textbf{Execution phase}: Execute the effectual operations as a batch
on the first slab, and then execute the residual operations as a batch
on the final slab, namely for each slab doing the following at each
section $S[k]$ in order (small to big):

\begin{enumerate}
\item Let $G_{1..4}$ be the partition of the batch of operations into the
$4$ access types (deletions last), each $G_{a}$ sorted by item.
\item For each segment $S_{i}[k]$ in $S[k]$, and for each $a$ from $1$
to $4$, cut out the operations that fit in $S_{i}[k]$ from $G_{a}$,
and perform those operations (as a sorted batch) on $S_{i}[k]$, and
then return their results.
\item Skip the rest of the slab if the batch becomes empty.
\end{enumerate}
\item \textbf{Rebalancing phase}: Rebalance all the segments and chains
by doing the following:

\begin{enumerate}
\item \label{enu:balance-seg} \textbf{\textit{Segment rebalancing}}: For
each chain $S_{i}$, for each segment $S_{i}[k]$ in $S_{i}$ in order
(small to big):

\begin{enumerate}
\item \label{enu:overflow} If $k>0$ and $S_{i}[k-1]$ is overfull, shift
items from $S_{i}[k-1]$ to $S_{i}[k]$ to make $S_{i}[k-1]$ have
target size.
\item \label{enu:fill} If $k>0$ and $S_{i}[k-1]$ is underfull and $S_{i}[k]$
either has at least $\frac{c(k)}{2}$ items or is the last segment
in $S_{i}$, let $S_{i}[k']$ be the first underfull segment in $S_{i}$,
and \textbf{fill} $S_{i}[k'..k-1]$ using $S_{i}[k]$ as follows:
for each $j$ from $k-1$ down to $k'$, shift items from $S_{i}[j+1]$
to $S_{i}[j]$ to make $S_{i}[k'..j]$ have total size $\sum_{a=k'}^{j}t(a)$
or as close as possible, and then remove $S_{i}[j+1]$ if it is emptied.

\item If $S_{i}[k]$ is (still) overfull and is the last segment in $S_{i}$,
create a new (empty) segment $S_{i}[k+1]$.
\item Skip the rest of the current slab if $S_{i}[k]$ is (now) balanced
and the execution phase had skipped $S[k]$.
\end{enumerate}
\item \label{enu:balance-chain} \textbf{\textit{Chain rebalancing}}: After
that, if one chain $S_{i}$ is longer than the other chain $S_{j}$,
repeat the following until the chains are the same length:

\begin{enumerate}
\item \label{enu:transfer} Let the current chains be $S_{i}[0..k]$ and
$S_{j}[0..k']$. Create new (empty) segments $S_{j}[k'+1..k]$, and
shift all items from $S_{i}[k]$ to $S_{j}[k]$, and then \textbf{fill}
the underfull segments in $S_{j}[k'..k-1]$ using $S_{j}[k]$ (as
in \step\ref{enu:fill}).
\item If $S_{j}[k]$ is (now) empty again, remove $S[k]$.
\end{enumerate}
\end{enumerate}
\end{enumerate}

\subsection{Analysis of $\protect\fs_{1}$}

First we establish that the rebalancing phase works, by proving the
following two lemmas.
\begin{lem}[$\fs_{1}$ Segment Rebalancing Invariant]
\label{lem:FS1-seg-rebalancing} During the segment rebalancing (\step\ref{enu:balance-seg}),
just after the iteration for segment $S_{i}[k]$, for any imbalanced
segment $S_{i}[k']$ in $S_{i}[0..k]$, either $k'=k$ or $S_{i}[k'..k]$
are all underfull.\end{lem}
\begin{proof}
The invariant clearly holds for $S_{i}[0]$. Consider each iteration
for segment $S_{i}[k]$ during the segment rebalancing where $k>0$.
If $S_{i}[k-1]$ was overfull, then by the invariant it was the only
imbalanced segment in $S_{i}[0..k-1]$, and would be rebalanced in
\step\ref{enu:overflow}, preserving the invariant. If $S_{i}[k-1]$
was underfull and $S_{i}[k]$ had at least $\frac{c(k)}{2}$ items
or was the last segment in $S_{i}$, then in \step\ref{enu:fill}
$S_{i}[k'..k-1]$ would be filled using $S_{i}[k]$, which had at
least $\frac{c(k)}{2}\ge\sum_{a=k'}^{k-1}t(a)$ items unless it was
the last segment in $S_{i}$, and hence after that every segment in
$S_{i}[k'..k-1]$ (that is not removed) would be balanced, preserving
the invariant. If \step\ref{enu:overflow} and \step\ref{enu:fill}
do not apply, then $S_{i}[k-1]$ is balanced or $S_{i}[k]$ is underfull,
so the invariant is preserved. Finally, if $S_{i}[k]$ is balanced
at the end of that iteration, and had been skipped by the execution
phase, then by the invariant all segments in $S_{i}[0..k]$ are balanced,
and all segments skipped by the rebalancing phase are also balanced,
so the invariant is preserved.
\end{proof}

\begin{lem}[$\fs_{1}$ Chain Rebalancing Iterations]
\label{lem:FS1-chain-rebalancing} The chain rebalancing (\step\ref{enu:balance-chain})
takes at most two iterations, after which both chains $S_{0}$ and
$S_{1}$ will have equal length and all their segments will be balanced.\end{lem}
\begin{proof}
By \ref{lem:FS1-seg-rebalancing}, all segments in each chain will
be balanced after the segment rebalancing (\step\ref{enu:balance-seg}).
After that, if one chain $S_{i}[0..k]$ is longer than the other chain
$S_{j}[0..k']$, the first chain rebalancing iteration transfers all
items in $S_{i}[k]$ to the other chain (\step\ref{enu:transfer}),
leaving $S_{i}[k]$ empty. If $S_{j}[k]$ remains non-empty, then
both chains have length $k$ and we are done. Otherwise, $S[k]$ would
be removed, and then the second chain rebalancing iteration transfers
all items in $S_{i}[k-1]$ to the other chain, which is at least $c(k-1)\ge\sum_{a=0}^{k-2}t(a)$
items, so every segment in $S_{j}[k'..k-2]$ would be filled to target
size, and hence both chains would have length $(k-1)$.
\end{proof}
Next we bound the work done by $\fs_{1}$.
\begin{defn}[Inward Order]
\label{def:inward-order} Take any sequence $A$ of map operations
and let $I$ be the set of items accessed by operations in $A$. Define
the \textbf{inward distance} of an operation in $A$ on an item $x$
to be $\min(\f{size}(I_{\le x}),\f{size}(I_{\ge x}))$. We say that
$A$ is in \textbf{inward order} iff its operations are in order of
(non-strict) increasing inward distance. Naturally, we say that $A$
is in \textbf{outward order} iff its reverse is in inward order.

\end{defn}
\begin{thm}[$\fs_{1}$ Work]
\label{thm:FS1-work} $\fs_{1}$ takes $O(F_{L})$ work for some
linearization $L$ of $\fs_{1}$-calls in $D$.\end{thm}
\begin{proof}
Let $L^{*}$ be a linearization of $\fs_{1}$-calls in $D$ such that:
\begin{itemize}
\item Operations on $\fs_{1}$ in earlier input batches are before those
in later input batches.
\item The operations within each batch are ordered as follows:

\begin{enumerate}
\item Ineffectual operations are before effectual/residual operations.
\item Effectual/residual operations are in order of access type (deletions
last).
\item Effectual insertions are in inward order, and \label{enu:FS1-lin-effect-del}
effectual deletions are in outward order.
\item Operations in each group-operation are consecutive and in the same
order as in that group.
\end{enumerate}
\end{itemize}
Let $L'$ be the same as $L^{*}$ except that in \point \ref{enu:FS1-lin-effect-del}
effectual deletions are ordered so that those on items in earlier
sections are later (instead of outward order). Now consider each input
batch $B$ of $b$ operations on $\fs_{1}$. 

In the preliminary and execution phases, each section $S[a]$ takes
$O\left(2^{a}\right)$ work per operation. Thus each operation in
$B$ with finger distance $r$ according to $L'$ on an item $x$
that was found to fit in section $S[k]$ takes $O\left(\sum_{a=0}^{k}2^{a}\right)=O\left(2^{k}\right)\wi O(\log r+1)$
work, because $r\ge\sum_{a=0}^{k-1}c(a)+1\ge\frac{1}{2}c(k-1)$ if
$S[k]$ is in the first slab (since earlier effectual operations in
$B$ did not delete items in $S[0..k-1]$), and $r\ge\sum_{a=0}^{k-1}c(a)-b\ge\frac{1}{2}c(k-1)$
if $S[k]$ is in the final slab (since $b\le\frac{1}{2}c(m-1)$).
Therefore these phases take $O(F_{L'})$ work in total.

Let $G$ be the effectual operations in $B$ as a subsequence of $L^{*}$.
Entropy-sorting $G$ takes $O(H+b)$ work (Appendix \ref{thm:par-esort-cost}),
where $H$ is the entropy of $G$ (i.e.~$H=\sum_{i=1}^{b}\log\frac{b}{q_{i}}$
where $q_{i}$ is the number of occurrences of the $i$-th operation
in $G$). Partition $G$ into $3$ parts: searches/updates $G_{1}$
and insertions $G_{2}$ and deletions $G_{3}$. And let $H_{j}$ be
the entropy of $G_{j}$. Then $H=\sum_{j=1}^{3}H_{j}+\sum_{i=1}^{b}\log\frac{b}{b_{i}}$
where $b_{i}$ is the number of operations in the same part of $G$
as the $i$-th operation in $G$, and $\sum_{i=1}^{b}\log\frac{b}{b_{i}}\le b\cdot\log\left(\frac{1}{b}\sum_{i=1}^{b}\frac{b}{b_{i}}\right)=b\cdot\log3$
by Jensen's inequality. Thus entropy-sorting $G$ takes $O\left(\sum_{j=1}^{3}H_{j}+b\right)$
work. Let $C_{j}$ be the cost of $G_{j}$ according to $F_{L^{*}}$.
Since each operation in $G_{j}$ has inward distance (with respect
to $G_{j}$) at most its finger distance according to $L^{*}$, we
have $H_{j}\in O(C_{j})$ (Appendix \ref{thm:max-finger}), and hence
entropy-sorting takes $O(F_{L^{*}})$ work in total.

Sorting the residual operations in the batch $B$ (that do not fit
in the first slab) takes $O(\log b)\wi O(\log r)$ work per operation
with finger distance $r$ according to $L^{*}$, since $r\ge c(m-1)\ge2b$.

Therefore the separation phase takes $O(F_{L^{*}})$ work in total.
Finally, the rebalancing phase takes $O(1)$ amortized work per operation,
as we shall prove in the next lemma. Thus $\fs_{1}$ takes $O(\max(F_{L^{*}},F_{L'}))$
total work.\end{proof}
\begin{lem}[$\fs_{1}$ Rebalancing Work]
\label{lem:FS1-rebalance-work} The rebalancing phase of $\fs_{1}$
takes $O(1)$ amortized work per operation.\end{lem}
\begin{proof}
We shall maintain the credit invariant that each segment $S_{i}[k]$
with $q$ items beyond its target capacity has at least $q\cdot2^{-k}$
stored credits. The execution phase clearly increases the total stored
credits needed by at most $1$ per operation, which we can pay for.
We now show that the invariant can be preserved after the segment
rebalancing and the chain rebalancing.

During the segment rebalancing (\step\ref{enu:balance-seg}), each
shift is performed between some neighbouring segments $S_{i}[k]$
and $S_{i}[k+1]$, where $S_{i}[k]$ has $t(k)+q$ items and $S_{i}[k+1]$
has $t(k+1)+q'$ items just before the shift, and $|q|>c(k)$. The
shift clearly takes $O(\log(t(k)+q)+\log(t(k+1)+q'))$ work. If $q'<2\cdot t(k+1)$
then this is obviously just $O(\log t(k)+\log|q|)$ work. But if $q'>2\cdot t(k+1)$,
then $S_{i}[k+1]$ will also be rebalanced in \step\ref{enu:overflow}
of the next segment balancing iteration, since at most $\sum_{a=0}^{k}t(a)\le t(k+1)$
items will be shifted from $S_{i}[k+1]$ to $S_{i}[k]$ in \step\ref{enu:fill},
and hence $S_{i}[k+1]$ will still have at least $q'$ items. In that
case, the second term $O(\log(t(k+1)+q')))$ in the work bound for
this shift can be bounded by the first term of the work bound for
the subsequent shift from $S_{i}[k+1]$ to $S_{i}[k+2]$, since $\log(t(k+1)+q')\in O(\log q')$.
Therefore in any case we can treat this shift as taking only $O(\log t(k)+\log|q|)\wi O(\log|q|)\wi O\left(|q|\cdot2^{-k}\right)$
work.

Now consider the two kinds of segment rebalancing:
\begin{itemize}
\item \textit{Overflow}: \step\ref{enu:overflow} shifts items from overfull
$S_{i}[k]$ to $S_{i}[k+1]$. Suppose that $S_{i}[k]$ has $t(k)+u$
items just before the shift. After the shift, $S_{i}[k]$ has target
size and needs no stored credits, and $S_{i}[k+1]$ would need at
most $u\cdot2^{-(k+1)}$ extra stored credits. Thus the $u\cdot2^{-k}$
credits stored at $S_{i}[k]$ can pay for both the shift and the needed
extra stored credits.
\item \textit{Fill}: \step\ref{enu:fill} fills some underfull segments
$S_{i}[k'..k]$ using $S_{i}[k+1]$. Suppose that $S_{i}[j]$ has
$t(j)-u_{i}(j)$ items just before the fill, for each $j\in[k'..k]$.
After the fill, every segment in $S_{i}[k'..k]$ has size within target
capacity and needs no stored credits, and $S_{i}[k+1]$ needs at most
$\left(\sum_{j=k'}^{k}u_{i}(j)\right)\cdot2^{-(k+1)}\le\frac{1}{2}\sum_{j=k'}^{k}\left(u_{i}(j)\cdot2^{-j}\right)$
extra stored credits, which can be paid for by using half the credits
stored at each segment in $S_{i}[k'..k]$. The other half of the $u_{i}(j)\cdot2^{-j}$
credits stored at $S_{i}[j]$ suffices to pay for the shift from $S_{i}[j+1]$
to $S_{i}[j]$, for each $j\in[k'..k]$.
\end{itemize}
The chain rebalancing (\step\ref{enu:balance-chain}) is performed
only when segment rebalancing creates or removes a segment and makes
one chain longer than the other. Consider the biggest segment $S_{i}[k]$
that was created or removed. If $S_{i}[k]$ was created, it must be
due to overflowing $S_{i}[k-1]$ to $S_{i}[k]$ in \step\ref{enu:overflow},
and hence the shift from $S_{i}[k-1]$ to $S_{i}[k]$ already took
$\Theta\left(2^{k}\right)$ work. If $S_{i}[k]$ was removed, it must
be due to filling some segments $S_{i}[k'..k-1]$ using $S_{i}[k]$
in \step\ref{enu:fill}, but $S_{i}[k-1]$ must have had at least
$c(k-1)$ items before the execution phase, and at least half of them
were either deleted or shifted to $S_{i}[k-2]$, and hence either
the deletions can pay $\Theta\left(2^{k}\right)$ credits, or the
shift to $S_{i}[k-2]$ already took $\Theta\left(2^{k}\right)$ work.
Therefore in any case we can afford to ignore up to $\Theta\left(2^{k}\right)$
work done by chain rebalancing.

Now observe that the chain rebalancing performs at most two transfers
(\step\ref{enu:transfer}) of items from the last segment of the
longer chain $S_{i}[0..k]$ to the shorter chain $S_{i'}[0..k']$,
by the \nameref{lem:FS1-chain-rebalancing} (\ref{lem:FS1-chain-rebalancing}).
Each transfer takes $O(k)$ work to create the new segments and $O(1)$
work to shift $S_{i}[k]$ over to $S_{i'}[k]$, and then fills underfull
segments in $S_{i'}[k'..k-1]$ using $S_{i'}[k]$. The fill takes
$O\left(2^{k}\right)$ work for the shift from $S_{i'}[k]$ to $S_{i'}[k-1]$,
and takes $O\left(2^{j}\right)$ work for each shift from $S_{i'}[j]$
to $S_{i'}[j-1]$ for each $j\in[k'+1..k-1]$, since $S_{i'}[j]$
has at most $\sum_{a=0}^{j}t(a)\le t(j+1)$ items just before the
shift. Therefore each transfer takes $O\left(2^{k}\right)$ work in
total, and hence we can ignore all the work done by the chain rebalancing.
\end{proof}
And now we turn to bounding the span of $\fs_{1}$.
\begin{thm}[$\fs_{1}$ Span]
\label{thm:FS1-span} $\fs_{1}$ takes $O\left(\frac{N}{p}+d\cdot\left((\log p)^{2}+\log n\right)\right)$
span, where $N$ is the number of operations on $\fs_{1}$, and $n$
is the maximum size of $\fs_{1}$, and $d$ is the maximum number
of $\fs_{1}$-calls on any path in the program DAG $D$.\end{thm}
\begin{proof}
Let $s(b)$ denote the maximum span of processing an input batch of
size $b$ (that has been flushed from the parallel buffer). Take any
input batch $B$ of size $b$. We shall bound the span taken by $B$
in each phase.

The preliminary phase takes $O\left(\log b\cdot2^{k}\right)$ span
in each first slab segment $S[k]$, adding up to $O\left((\log b)^{2}\right)$
span. The separation phase also takes $O\left((\log b)^{2}\right)$
span, by \nameref{thm:par-esort-cost} (\ref{thm:par-esort-cost}).
The execution phase takes $O\left(\log b+2^{k}\right)$ span in each
segment $S[k]$, adding up to $O\left(\log b\cdot\log\log b+\log n\right)$
span. Returning the results for each group-operation takes $O(\log b)$
span.

The rebalancing phase also takes $O\left(\log b+2^{k}\right)$ span
for each segment $S[k]$ processed in \step\ref{enu:balance-seg},
because each shift between segments with total size $q$ takes $O(\log q)$
span, and filling $S_{i}[k'..k-1]$ using $S_{i}[k]$ in \step\ref{enu:fill}
takes $O\left(\log\left(b+\sum_{a=k'}^{k}t(a)\right)\right)\wi O\left(\log b+2^{k}\right)$
span for the first shift from $S_{i}[k]$ to $S_{i}[k-1]$ and then
$O\left(\log\sum_{a=k'}^{j}t(a)\right)\wi O\left(2^{j}\right)$ span
for each subsequent shift from $S_{i}[j+1]$ to $S_{i}[j]$. Similarly,
the chain rebalancing in \step\ref{enu:balance-chain} takes $O(\log b+\log n)$
span, because it performs at most two iterations by \nameref{lem:FS1-chain-rebalancing}
(\ref{lem:FS1-chain-rebalancing}), each of which takes $O(\log b+\log n)$
span to fill the underfull segments of the shorter chain using its
last segment.

Therefore $s(b)\in O\left((\log b)^{2}+\log n\right)\wi O\left(\frac{b}{p}+(\log p)^{2}+\log n\right)$,
since $(\log b)^{2}\in O\left(\frac{b}{p}\right)$ if $b\ge p^{2}$.

Each batch $B$ of size $b$ waits in the buffer for the preceding
batch of size $b'$ to be processed, taking $O\left(s(b')\right)$
span, and then $B$ itself is processed, taking $O(s(b))$ span, taking
$O(s(b)+s(b'))$ span in total. Since over all batches each of $b,b'$
will sum up to at most the total number $N$ of $\fs_{1}$-calls,
and there are at most $d$ $\fs_{1}$-calls on any path in the program
DAG $D$, the span of $\fs_{1}$ is $O\left(\frac{N}{p}+d\cdot\left((\log p)^{2}+\log n\right)\right)$.
\end{proof}

\section{Faster Parallel Finger Structure}

\label{sec:PFS2}

Although $\fs_{1}$ has optimal work and a small span, it is possible
to reduce the span even further, intuitively by pipelining the batches
in some fashion so that an expensive access in a batch does not hold
up the next batch.

As with $\fs_{1}$, we need to split the sections into two slabs,
but this time we fix the first slab at $m$ sections where $m\in\log\Theta(\log p)$
so that we can pipeline just the final slab. We need to allow big
enough batches so that operations that are delayed because earlier
batches are full can count their delay against the total work divided
by $p$. But to keep the span of the sorting phase down to $O\left((\log p)^{2}\right)$,
we need to restrict the batch size. It turns out that restricting
to batches of size at most $p^{2}$ works.

We cannot pipeline the first slab (particularly the rebalancing),
but the preliminary phase and separation phase would only take $O\left((\log p)^{2}\right)$
span. The execution phase and rebalancing phases are still carried
out as before on the first slab, taking $O\left((\log p)^{2}\right)$
span, but execution and rebalancing on the final slab are pipelined,
by having each final slab section $S[k]$ process the batch passed
to it and rebalance the preceding segments $S_{0}[k-1]$ and $S_{1}[k-1]$
if necessary.

To guarantee that this local rebalancing is possible, we do not allow
$S[k]$ to proceed if it is imbalanced or if there are more than $c(k)$
pending operations in the buffer to $S[k+1]$. In such a situation,
$S[k]$ must stop and reactivate $S[k+1]$, which would clear its
buffer and rebalance $S[k]$ before restarting $S[k]$. It may be
that $S[k+1]$ also cannot proceed for the same reason and is stopped
in the same manner, and so $S[k]$ may be delayed by such a stop for
a long time. But by a suitable accounting argument we can bound the
total delay due to all such stops by the total work divided by $p$.
Similarly, we do not allow the first slab to run (on a new batch)
if $S[m-1]$ is imbalanced or there are more than $c(m-1)$ pending
operations in the buffer to $S[m]$.

Finally, we use an odd-even locking scheme to ensure that the segments
in the final slab do not interfere with each other yet can proceed
at a consistent pace. The details are below.

\subsection{Description of $\protect\fs_{2}$}

\label{sub:FS2-desc}

\begin{figure}[H]
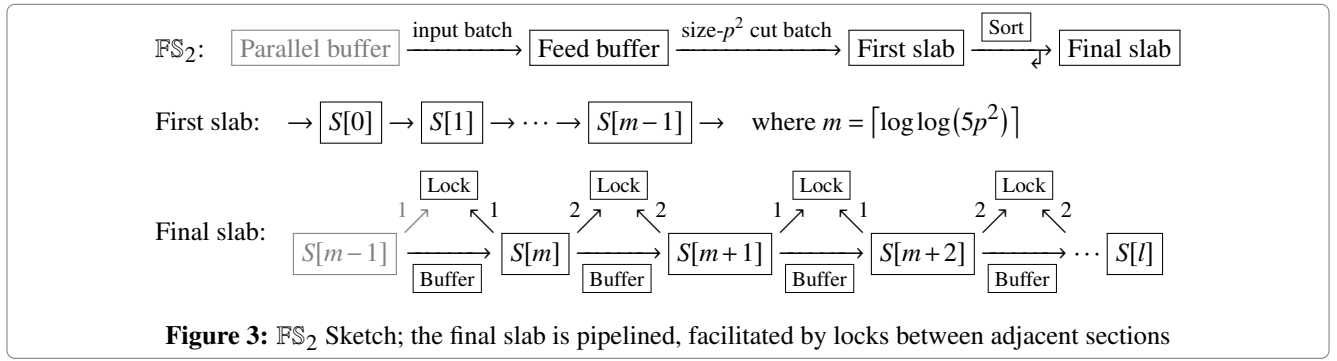

\begin{centerbox}
\begin{roundedboxinfloat}
\begin{centerbox}
\noindent %
\begin{tabular}{l}
$\fs_{2}$:\quad{}$\grey{\tbox{Parallel\ buffer}}\xrightarrow{\text{input batch}}\tbox{Feed\ buffer}\xrightarrow{\text{size-\ensuremath{p^{2}} cut batch}}\tbox{First\ slab}\xrightarrow[\smash{\qquad\raisebox{.84em}{\tilt{\Rsh}{180}}}]{\stbox{Sort}\ }\tbox{Final\ slab}$\tabularnewline
\noalign{\vskip1em}
First slab:\quad{}$\to\boxed{S[0]}\to\boxed{S[1]}\to\cdots\to\boxed{S[m-1]}\to$\quad{}where
$m=\ceil{\log\log\left(5p^{2}\right)}$\tabularnewline
\noalign{\vskip1em}
Final slab:\quad{}$\begin{matrix}\vsp{.4}\\
\grey{\boxed{S[m-1]}}
\end{matrix}\begin{matrix}\stbox{Lock}\\
\grey{\mathrlap{^{1}}\nearrow}\quad\ \nwarrow\mathllap{^{1}}\\
{}\xrightarrow[\stbox{Buffer}]{}{}
\end{matrix}\begin{matrix}\vsp{.4}\\
\boxed{S[m]}
\end{matrix}\begin{matrix}\stbox{Lock}\\
\mathrlap{^{2}}\nearrow\quad\ \nwarrow\mathllap{^{2}}\\
{}\xrightarrow[\stbox{Buffer}]{}{}
\end{matrix}\begin{matrix}\vsp{.4}\\
\boxed{S[m+1]}
\end{matrix}\begin{matrix}\stbox{Lock}\\
\mathrlap{^{1}}\nearrow\quad\ \nwarrow\mathllap{^{1}}\\
{}\xrightarrow[\stbox{Buffer}]{}{}
\end{matrix}\begin{matrix}\vsp{.4}\\
\boxed{S[m+2]}
\end{matrix}\begin{matrix}\stbox{Lock}\\
\mathrlap{^{2}}\nearrow\quad\ \nwarrow\mathllap{^{2}}\\
{}\xrightarrow[\stbox{Buffer}]{}{}
\end{matrix}\begin{matrix}\vsp{.4}\\
\cdots\,\boxed{S[l]}
\end{matrix}$\tabularnewline
\noalign{\vskip0.5em}
\end{tabular}
\end{centerbox}
\caption{\label{fig:FS2-flow}$\protect\fs_{2}$ Sketch; the final slab is
pipelined, facilitated by locks between adjacent sections}
\end{roundedboxinfloat}
\end{centerbox}
\end{figure}

We shall now give the details (see \ref{fig:FS2-flow}). We will need
the \textbf{bunch} structure (Appendix \ref{def:bunch}) for aggregating
batches, which is an unsorted set supporting both addition of a batch
of new elements within $O(1)$ work/span and conversion to a batch
within $O(b)$ work and $O(\log b)$ span if it has size $b$.

$\fs_{2}$ has the same sections as in $\fs_{1}$, with the \textbf{first
slab} comprising the first $m=\ceil{\log\log\left(5p^{2}\right)}$
sections, and the \textbf{final slab} comprising the other sections.
$\fs_{2}$ uses a \textbf{feed buffer}, which is a queue of bunches
of operations each of size exactly $p^{2}$ except the last (which
can be empty). Whenever $\fs_{2}$ is notified of input (by the parallel
buffer), it reactivates the first slab.

Each section $S[k]$ in the final slab has a \textbf{buffer} before
it (for pending operations from $S[k-1]$), which for each access
type uses an optimal batch-parallel map (Appendix \ref{sub:batch-map})
to store bunches of group-operations of that type, where operations
on the same item are in the same bunch. When a batch of group-operations
on an item is inserted into the buffer, it is simply added to the
correct bunch. Whenever we count operations in the buffer, we shall
count them individually even if they are on the same item. The first
slab and each final slab section also has a \textbf{deferred flag},
which indicates whether its run is deferred until the next section
has run. Between every pair of consecutive sections starting from
after $S[m-1]$ is a \textbf{neighbour-lock}, which is a dedicated
lock (see \ref{sub:primitives}) with $1$ key for each arrow to it
in \ref{fig:FS2-flow}.

Whenever the first slab is reactivated, it runs as follows:
\begin{enumerate}
\item If the parallel buffer and feed buffer are both empty, terminate.
\item Acquire the neighbour-lock between $S[m-1]$ and $S[m]$. (Skip steps
2 to 4 and steps 8 to 10 if $S[m]$ does not exist.)
\item If $S[m-1]$ has any imbalanced segment or $S[m]$ has more than $c(m-1)$
operations in its buffer, set the first slab's deferred flag and release
the neighbour-lock, and then reactivate $S[m]$ and terminate.
\item Release the neighbour-lock.
\item Let $q$ be the size of the last bunch $F$ in the feed buffer. Flush
the parallel buffer (if it is non-empty) and cut the input batch of
size $b$ into small batches of size $p^{2}$ except possibly the
first and last, where the first has size $\min\left(b,p^{2}-q\right)$.
Add that first small batch to $F$, and append the rest as bunches
to the feed buffer.
\item Remove the first bunch from the feed buffer and convert it into a
batch $B$, which we call a \textbf{cut batch}.
\item Process $B$ using the same four phases as in $\fs_{1}$ (\ref{par:FS1-phases}),
but restricted to the first slab (i.e.~execute only the effectual
operations on the first slab, and do segment rebalancing only on the
first slab, and do chain rebalancing only if $S[m]$ had not existed
before this processing). Furthermore, do not update $S[m-1]$'s segments'
sizes until after this processing (so that $S[m]$ in \step\ref{enu:FS2-rebalance}
will not find any of $S[m-1]$'s segments imbalanced until the first
slab rebalancing phase has finished).
\item Acquire the neighbour-lock between $S[m-1]$ and $S[m]$.
\item Insert the residual group-operations (on items that do not fit in
the first slab) into the buffer of $S[m]$, and then reactivate $S[m]$.
\item Release the neighbour-lock.
\item Reactivate itself.
\end{enumerate}
Whenever a final slab section $S[k]$ is reactivated, it runs as follows:
\phantomsection\label{par:FS2-final-slab-run}
\begin{enumerate}
\item Acquire the neighbour-locks (between $S[k]$ and its neighbours) in
the order given by the arrow number in \ref{fig:FS2-flow}.
\item If $S[k]$ has any imbalanced segment or $S[k+1]$ (exists and) has
more than $c(k)$ operations in its buffer, set $S[k]$'s deferred
flag and release the neighbour-locks, and then reactivate $S[k+1]$
and terminate.
\item For each access type, flush and process the (sorted) batch $G$ of
bunches of group-operations of that type in its buffer as follows:

\begin{enumerate}
\item Convert each bunch in $G$ to a batch of group-operations.
\item For each segment $S_{i}[k]$ in $S[k]$, cut out the group-operations
on items that fit in $S_{i}[k]$ from $G$, and perform them (as a
sorted batch) on $S_{i}[k]$, and then fork to return the results
of the operations (according to the order within each group-operation).
\item If $G$ is non-empty (i.e.~has leftover group-operations), insert
$G$ into the buffer of $S[k+1]$ and then reactivate $S[k+1]$.
\end{enumerate}
\item \label{enu:FS2-rebalance} Rebalance locally as follows (essentially
like in $\fs_{1}$):

\begin{enumerate}
\item \label{enu:FS2-balance-seg} For each segment $S_{i}[k]$ in $S[k]$:

\begin{enumerate}
\item If $S_{i}[k-1]$ is overfull, shift items from $S_{i}[k-1]$ to $S_{i}[k]$
to make $S_{i}[k-1]$ have target size.
\item If $S_{i}[k-1]$ is underfull, shift items from $S_{i}[k]$ to $S_{i}[k-1]$
to make $S_{i}[k-1]$ have target size, and then remove $S_{i}[k]$
if it is emptied.
\item If $S_{i}[k]$ is (still) overfull and is the last segment in $S_{i}$,
create a new segment $S_{i}[k+1]$ and reactivate it.
\end{enumerate}
\item \label{enu:FS2-balance-chain} If $S[k]$ is (still) the last section,
but chain $S_{i}$ is longer than chain $S_{j}$:

\begin{enumerate}
\item Create a new segment $S_{j}[k]$ and shift all items from $S_{i}[k]$
to $S_{j}[k]$.
\item If $S_{j}[k-1]$ is (now) underfull, shift items from $S_{j}[k]$
to $S_{j}[k-1]$ to make $S_{j}[k-1]$ have target size.
\item If $S_{j}[k]$ is (now) empty again, remove $S[k]$.
\end{enumerate}
\end{enumerate}
\item If $k=m$, and the first slab is deferred, clear its deferred flag
then reactivate it.
\item If $k>m$, and $S[k-1]$ is deferred, clear its defered flag then
reactivate it.
\item Release the neighbour-locks.
\end{enumerate}

\subsection{Analysis of $\protect\fs_{2}$}

For each computation, we shall define its delay to intuitively capture
the minimum time it needs, including all potential waiting on locks.
Each blocked acquire of a dedicated lock corresponds to an \textbf{acquire-stall
node} $\alpha$ in the execution DAG whose child node $\rho$ is created
by the release just before the successful acquisition of the lock.
Let $\boldsymbol{\Delta}(\alpha)$ be the ancestor nodes of $\rho$
that have not yet executed at the point when $\alpha$ is executed.
Then the \textbf{delay} of a computation $\Gamma$ is recursively
defined as the weighted span of $\Gamma$, where each acquire-stall
node $\alpha$ in $\Gamma$ is weighted by the delay of $\Delta(\alpha)$
(to capture the total waiting at $\alpha$), and every other node
is weighted by its cost.~\footnote{The delay of $\Gamma$ depends on the actual execution, due to the
definition of $\Delta(\alpha)$ for each acquire-stall node $\alpha$
in $\Gamma$. But it captures the minimum time needed to run $\Gamma$
in the following sense: For any computation $\Gamma$, on any step
that executes all ready nodes in the remaining computation $\Gamma'$
(i.e.~the unexecuted nodes in $\Gamma$), the delay of $\Gamma'$
is reduced. (So if a greedy scheduler is used, the number of steps
in which some processor is idle is bounded by the delay.)}

Whenever the first slab or a final slab section runs, we say that
it \textbf{defers} if it terminates with its deferred flag set (i.e.~at
step 2), otherwise we say that it \textbf{proceeds} (i.e.~to step
3) and eventually \textbf{finishes} (i.e.~reaches the last step)
with its deferred flag cleared. We now establish some invariants,
which guarantee that $\fs_{2}$ is always sufficiently balanced.
\begin{lem}[$\fs_{2}$ Balance Invariants]
\label{lem:FS2-balance-inv} $\fs_{2}$ satisfies the following invariants:
\begin{enumerate}
\item When the first slab is not running, every segment in $S_{i}[0..m-2]$
is balanced and $S_{i}[m-1]$ has at most $2\cdot t(m-1)$ items.
\item When a final slab section $S[k]$ rebalances a segment in $S[k-1]$
(in \step\ref{enu:FS2-balance-seg}), it will make that segment have
size $t(k-1)$.
\item Just after the last section finishes without creating new sections,
the segments in $S[k]$ are balanced and both chains have the same
length.
\item Each final slab section $S[k]$ always has at most $2\cdot c(k-1)$
operations in its buffer.
\item Each final slab segment $S_{i}[k]$ always has at most $2\cdot t(k)$
items, and at least $c(k-1)$ items unless $S[k]$ is the last section.
\end{enumerate}
\end{lem}
\begin{proof}
Invariant~1 holds as follows: The first slab proceeds only if $S[m-1]$'s
segments are balanced, and from that point until after the rebalancing
phase, its segments are modified only by itself (since $S[m]$ will
not modify $S[m-1]$), and thereafter all its sections except $S[m-1]$
remain unmodified until it processes the next cut batch. Thus the
same proof as for \nameref{lem:FS1-seg-rebalancing} (\ref{lem:FS1-seg-rebalancing})
shows that just before the segment rebalancing (\step\ref{enu:balance-seg})
iteration for $S_{i}[m-1]$, for any imbalanced first slab segment
$S_{i}[k]$, either $k=m-2$ or $S_{i}[k..m-2]$ are underfull. But
note that the cut batch had at most $p^{2}\le\frac{c(m-1)}{2}$ operations,
and so after the execution phase, $S_{i}[m-1]$ had at least $\frac{c(m-1)}{2}$
items unless it was the last segment in its chain. Thus $S_{i}[0..m-2]$
will be made balanced (by \step\ref{enu:overflow} or \step\ref{enu:fill}
in the iteration for $S_{i}[m-1]$, or by \step\ref{enu:balance-chain}).
Similarly, $S[m-1]$ will have at most $t(m-1)+\sum_{a=0}^{m-1}c(a)+p^{2}\le2\cdot t(m-1)$
items in each segment, since $\sum_{a=0}^{m-2}c(a)\le\frac{c(m-1)}{2}$.

Invariant~2 holds as follows. Each final slab section $S[k]$ proceeds
only if its segments each has at least $c(k)$ items unless it is
the last segment in its chain, and its buffer had at most $2\cdot c(k-1)$
operations by Invariant~4. Since $c(k)-2\cdot c(k-1)\ge t(k-1)$,
rebalancing a segment in $S[k-1]$ (\step\ref{enu:FS2-balance-seg})
will make it have size $t(k-1)$.

Invariant~3 holds as follows. The last section $S[l]$ proceeds only
if its segments each has at most $t(k)+c(k)$ items, and its buffer
had at most $2\cdot c(k-1)\le c(k)$ operations by Invariant~4. Thus
if any of its segments $S_{i}[l]$ becomes overfull and it creates
a new section $S[l+1]$, it will subsequently be deferred until $S[l+1]$
runs. And during that run of $S[l+1]$, it will proceed and shift
at most $2\cdot c(k)$ items from $S_{i}[l]$ to $S_{i}[l+1]$, after
which $S_{i}[l+1]$ will not be overfull, and so $S[l+1]$ will not
create another new section $S[l+2]$. Therefore we can assume that
the chains' lengths never differ by more than one segment, and so
the chain rebalancing (\step\ref{enu:FS2-balance-chain}) will make
the chains the same length while ensuring the segments in $S[k-1]$
and $S[k]$ are balanced.

Invariant~4 holds for $S[m]$, because the first slab proceeds only
if $S[m]$'s buffer has at most $c(m-1)$ operations, and only processes
a cut batch of size at most $p^{2}$, hence after that $S[m]$'s buffer
will have at most $p^{2}+c(m-1)<2\cdot c(m-1)$ operations. Invariant~4
holds for $S[k]$ for each $k>m$, because $S[k-1]$ proceeds only
if $S[k]$'s buffer has at most $c(k-1)$ operations, and only processes
a buffered batch of size at most $2\cdot c(k-2)$ by Invariant~4
for $S[k-1]$, hence after that $S[k]$'s buffer will have at most
$c(k-1)+2\cdot c(k-2)\le2\cdot c(k-1)$ operations.

Invariant~5 holds as follows. Each final slab segment $S_{i}[k]$
is modified only when either $S[k]$ or $S[k+1]$ runs, and the latter
never makes $S_{i}[k]$ imbalanced. Consider each $S[k]$ run. It
proceeds only if $S_{i}[k]$ has at most $t(k)+c(k)$ items and at
least $c(k)$ items unless $S[k]$ is the last section, and its buffer
had at most $2\cdot c(k-1)$ operations by Invariant~4, and $S_{i}[k-1]$
had at most $2\cdot t(k-1)$ items by Invariant~5 for $S[k-1]$.
So at most $2\cdot c(k-1)$ items were inserted into $S_{i}[k]$,
and at most $t(k-1)$ items were shifted from $S_{i}[k-1]$ to $S_{i}[k]$.
Also, at most $2\cdot c(k-1)$ items were deleted from $S_{i}[k]$,
and at most $t(k-1)$ items were shifted from $S_{i}[k]$ to $S_{i}[k-1]$.
Thus after that run, $S_{i}[k]$ has at most $t(k)+c(k)+4\cdot c(k-1)\le2\cdot t(k)$
items and at least $c(k)-2\cdot c(k-1)-t(k-1)\ge c(k-1)$ items unless
$S[k]$ was the last section, since $c(k)=c(k-1)^{2}\ge c(m-1)\cdot c(k-1)\ge5\cdot c(k-1)$.
\end{proof}
With these invariants, we are ready to bound the work done by $\fs_{2}$.
\begin{thm}[$\fs_{2}$ Work]
\label{thm:FS2-work} $\fs_{2}$ takes $O(F_{L})$ work for some
linearization $L$ of $\fs_{2}$-calls in $D$.\end{thm}
\begin{proof}
We shall use a similar proof outline as for \nameref{thm:FS1-work}
(\ref{thm:FS1-work}). Let $L^{*}$ be a linearization of $\fs_{2}$-calls
in $D$ such that:
\begin{itemize}
\item Operations on $\fs_{2}$ that finish during the first slab run or
some final slab section run are ordered by when that run finished.
\item Operations on $\fs_{2}$ that finish during the same first slab run
are ordered as follows:

\begin{enumerate}
\item Ineffectual operations are before effectual operations.
\item Effectual operations are in order of access type (deletions last).
\item Effectual insertions are in inward order, and \label{enu:FS2-lin-effect-del}
effectual deletions are in outward order (\ref{def:inward-order}).
\end{enumerate}
\item Operations on $\fs_{2}$ in each group-operation are in the same order
as in that group.
\end{itemize}
As before, let $L'$ be the same as $L^{*}$ except that in \point \ref{enu:FS2-lin-effect-del}
effectual deletions are ordered so that those on items in earlier
sections are later (instead of outward order).

Consider each cut batch $B$ of operations processed by the first
slab. By \nameref{lem:FS2-balance-inv} (\ref{lem:FS2-balance-inv}),
just before that processing, every segment in $S_{i}[0..m-2]$ is
balanced, and $S_{i}[m-1]$ has at most $2\cdot t(m-1)$ items. Thus
in both the preliminary phase and the execution phase, each section
$S[k]$ takes $O\left(2^{k}\right)$ work per operation. And this
amounts to $O(\log r+1)$ work per operation in $B$ with finger distance
$r$ according to $L'$, because the operation reaches $S[k]$ only
if $r\ge\sum_{a=0}^{\min(k-1,m-2)}c(k)+1$.

As with $\fs_{1}$, the separation phase takes $O(F_{L^{*}})$ work
in total (see \ref{thm:FS1-work}'s proof).

Now consider each batch $B$ of operations processed by a final slab
section $S[k]$. By \nameref{lem:FS2-balance-inv} (\ref{lem:FS2-balance-inv}),
$B$ has at most $2\cdot c(k-1)$ operations, and each segment in
$S[k]$ always has at most $4\cdot c(k)$ items. So inserting the
operations in $B$ into the buffer took $O\left(2^{k}\right)$ work
per operation. Converting each bunch in $B$ to a group-operation
takes $O(1)$ work per operation. Cutting out and performing and returning
the results of the group-operations that fit takes $O\left(2^{k}\right)$
work per group-operation. And the local rebalancing takes $O\left(2^{k}\right)$
work. Therefore each $S[k]$ run that proceeds to process its buffered
operations takes $O\left(2^{k}\right)$ work per operation. This again
amounts to $O(\log r+1)$ work per operation $X$ in $B$ with finger
distance $r$ according to $L^{*}$ as follows:
\begin{itemize}
\item If $X$ finishes in $S[m]$: At that point the first slab has at least
$c(m-1)-p^{2}\ge4p^{2}$ items in each chain, because $S[m-1]$ was
balanced just before processing the last cut batch. Thus $r\ge4p^{2}$
and hence $X$ costs $O\left(2^{m}\right)\wi O(\log r)$ work.
\item If $X$ finishes in $S[k]$ for some $k>m$: At that point $S[k-1]$
has at least $c(k-2)$ items in each segment by \nameref{lem:FS2-balance-inv}
(\ref{lem:FS2-balance-inv}). Thus $r\ge c(k-2)$ and hence $X$ costs
$O\left(2^{k}\right)\wi O(\log r)$ work.
\end{itemize}
Finally, all the rebalancing takes $O(1)$ amortized work per operation,
which we shall leave to the next lemma.\end{proof}
\begin{lem}[$\fs_{2}$ Rebalancing Work]
\label{lem:FS2-rebalance-work} All the rebalancing steps of $\fs_{2}$
take $O(1)$ amortized work per operation.\end{lem}
\begin{proof}
We shall maintain the credit invariant that each segment $S_{i}[k]$
with $q$ items beyond its target capacity has at least $q\cdot2^{-k}$
stored credits. Also, each unfinished operation carries $1$ credit
with it. As with $\fs_{1}$ (see \ref{lem:FS1-rebalance-work}'s proof),
the invariant can be preserved after rebalancing in the first slab.
By the same reasoning, the invariant can also be preserved after segment
rebalancing in the final slab (\step\ref{enu:FS2-balance-seg}),
because any shift between segments $S_{i}[k-1]$ and $S_{i}[k]$ where
$k\ge m$ is performed only when $S_{i}[k-1]$ is imbalanced, and
after that $S_{i}[k-1]$ has size $t(k-1)$ by \nameref{lem:FS2-balance-inv}
(\ref{lem:FS2-balance-inv}). Similarly, the invariant can be preserved
after chain rebalancing in the final slab (\step\ref{enu:FS2-balance-chain}),
because it takes $O\left(2^{k}\right)$ work, which can be ignored
since the last segment rebalancing shift already took $O\left(2^{k}\right)$
work.
\end{proof}
To tackle the span of $\fs_{2}$, we need some lemmas concerning the
span of cutting the input batch and the delay in each slab.
\begin{lem}[$\fs_{2}$ Input Cutting Span]
\label{lem:FS2-input-cut-span} The first slab cuts an input batch
of size $b$ (i.e.~cutting it into small batches and storing them
in the feed buffer) within $O\left(\frac{b}{p}+\log p\right)$ span.\end{lem}
\begin{proof}
Cutting the input batch into small batches takes $O(\log b)$ span.
Adding them to the feed buffer takes $O(1+\frac{b}{p^{2}})$ span.
This amounts to $O\left(\frac{b}{p}+\log p\right)$ span because $\log b\in O\left(\frac{b}{p}\right)$
if $b>p^{2}$.
\end{proof}
\clearpage{}
\begin{lem}[$\fs_{2}$ Final Slab Delay]
\label{lem:FS2-final-slab-delay} Each section $S[k]$ in the final
slab runs within $O\left(2^{k}\right)$ delay (whether it defers or
finishes).\end{lem}
\begin{proof}
Consider any final slab section $S[k]$ that has acquired the second
neighbour-lock. Checking whether it has an imbalanced segment and
checking $S[k+1]$'s buffer size takes only $O(1)$ delay. By \nameref{lem:FS2-balance-inv}
(\ref{lem:FS2-balance-inv}), $S[k]$ has at most $2\cdot c(k-1)$
operations in its buffer, and $S[k]$ always has at most $2\cdot t(k)$
items in each segment, and $S[k-1]$ has at most $2\cdot t(k-1)$
items in each segment. Thus converting each bunch in the buffer takes
$O\left(2^{k}\right)$ span, and performing the operations that fit
in $S[k]$ takes $O\left(2^{k}\right)$ span, and rebalancing the
segments in $S[k-1]$ takes $O\left(2^{k}\right)$ span.

Now consider any final slab section $S[k]$ that has acquired the
first neighbour-lock. It waits $O\left(2^{k}\right)$ delay for the
current holder (if any) of the second neighbour-lock to release it,
and then itself takes $O\left(2^{k}\right)$ more delay to complete
its run.

Finally consider any final slab section $S[k]$ that starts running.
If $k=m$, it waits $O\left(2^{k}\right)$ delay for the first slab
to release the shared neighbour-lock, since the first slab takes only
$O\left(2^{m}\right)$ span on each access to $S[m-1]$. If $k>m$,
it waits $O\left(2^{k}\right)$ delay for the current holder of the
first neighbour-lock to release it, and then itself takes $O\left(2^{k}\right)$
more delay to complete its run.\end{proof}
\begin{lem}[$\fs_{2}$ First Slab Delay]
\label{lem:FS2-first-slab-delay} The first slab takes $O(\log p)$
delay for each acquiring of the neighbour-lock, and it processes each
cut batch within $O\left((\log p)^{2}\right)$ delay.\end{lem}
\begin{proof}
Each acquiring of the neighbour-lock takes $O\left(2^{m}\right)=O(\log p)$
delay by \nameref{lem:FS2-final-slab-delay} (\ref{lem:FS2-final-slab-delay}).
Checking whether $S[m-1]$ has an imbalanced segment and checking
$S[m]$'s buffer size takes only $O(1)$ delay. Obtaining the cut
batch (whose size is at most $p^{2}$) from the first bunch from the
feed buffer takes $O(\log p)$ delay. The four phases take $O\left((\log p)^{2}\right)$
delay in total, as in $\fs_{1}$ (see \ref{thm:FS1-span}). Inserting
the residual group-operations into the buffer of $S[m]$ takes $O\left(\log p+2^{m}\right)=O(\log p)$
delay, since $S[m]$'s buffer had at most $2\cdot c(m-1)$ items by
\nameref{lem:FS2-balance-inv} (\ref{lem:FS2-balance-inv}).
\end{proof}
With these lemmas, we can finally bound the span of $\fs_{2}$.
\begin{thm}[$\fs_{2}$ Span]
\label{thm:FS2-span} $\fs_{2}$ takes $O\left(\frac{N}{p}+d\cdot(\log p)^{2}+s_{L}\right)$
span for some linearization $L$ of $D$. ($d$ is the maximum number
of $\fs_{2}$-calls on any path in $D$, and $s_{L}$ is the weighted
span of $D$ with $\fs_{2}$-calls weighted according to $F_{L}$.)\end{thm}
\begin{proof}
Take any path $C$ through the program DAG $D$. Let $L$ be the linearization
in the proof of \nameref{thm:FS2-work} (\ref{thm:FS2-work}). Consider
any $\fs_{2}$-call $X$ along $C$ with finger distance $r$ according
to $L$. We shall trace the journey of $X$ from the parallel buffer
in an input batch to a cut batch and then through the slabs, and bound
the delay taken by $X$ relative to $\fs_{2}$, meaning that in the
computation of the delay we only count $\fs_{2}$-nodes. Along the
way, we shall partition that delay into the \textbf{normal} delay
and the \textbf{deferment} delay, where the latter comprises all waiting
at the first slab or a section that defers from the point it sets
the deferred flag until it is reactivated and clears the deferred
flag (and proceeds).

\uline{Normal delay}

At the start, $X$ waits in the parallel buffer for the first slab
to finish running on the previous input batch of size $b'$, taking
$O\left(\frac{b'}{p}+(\log p)^{2}\right)$ delay by \nameref{lem:FS2-input-cut-span}
(\ref{lem:FS2-input-cut-span}) and \nameref{lem:FS2-first-slab-delay}
(\ref{lem:FS2-first-slab-delay}). Next $X$ waits for the first slab
to process some $i$ cut batches of size $p^{2}$ in the feed buffer,
each taking $O\left((\log p)^{2}\right)\wi O(p)$ normal delay. Then
$X$ is flushed from the parallel buffer in some input batch of size
$b$, which is cut within $O\left(\frac{b}{p}+(\log p)^{2}\right)$
normal delay, and next waits for another $j$ cut batches of size
$p^{2}$ that come before $X$ in the feed buffer, each taking $O(p)$
normal delay. (Note that we are ignoring all waiting while the first
slab is deferred.)

If $X$ finishes in the final slab, it waits a further $O\left(2^{k}\right)$
normal delay at each final slab section $S[k]$ that it passes through
by \nameref{lem:FS2-final-slab-delay} (\ref{lem:FS2-final-slab-delay}).
And when $X$ finishes in a section $S[k]$, at that point $S[k-1]$
has at least $c(k-2)$ items in each segment by \nameref{lem:FS2-balance-inv}
(\ref{lem:FS2-balance-inv}). Thus $r\ge c(k-2)$ and hence $X$ takes
$O\left(\sum_{a=m}^{k}2^{a}\right)=O\left(2^{k}\right)\wi O(\log r)$
normal delay in the final slab. Finally when $X$ is returned, it
is in some group-operation with $g$ operations, so returning the
results takes $O(\log g)\wi O\left(\frac{g}{p}+\log p\right)$ span.

Therefore in total $X$ takes $O\left(\frac{b'}{p}+\frac{b}{p}+\frac{g}{p}+i\cdot p+j\cdot p+(\log p)^{2}+\log r\right)$
normal delay.

\clearpage{}

\uline{Deferment delay}

To bound the deferment delay, we shall use a similar credit invariant
as in \nameref{lem:FS2-rebalance-work} (\ref{lem:FS2-rebalance-work}),
but instead of paying for rebalancing work we shall use the credits
to pay for $p$ times the deferment delay. This would imply that the
deferment delay is at most $O\left(\frac{1}{p}\right)$ per operation
on $\fs_{2}$. The invariant is that for $k\ge m-1$, each segment
$S_{i}[k]$ with $q$ items beyond its target capacity has at least
$q\cdot2^{-k}$ stored credits, and that each operation in $S[k+1]$'s
buffer carries $2^{-k}$ credits with it.

Consider each deferment of a section $S[k]$ for $k\ge m-1$ (where
deferment of the first slab is treated as deferment of $S[m-1]$).
At that point either one of its segment is imbalanced or $S[k+1]$'s
buffer has more than $c(k)$ items, and $S[k]$ reactivates $S[k+1]$,
which may either defer or proceed. In any case, from that point until
$S[k+1]$ proceeds, $S[k]$ will never proceed (even if reactivated),
because its segments and $S[k+1]$'s buffer remain untouched. But
once $S[k+1]$ proceeds, it will empty its buffer and make $S[k]$'s
segments balanced by \nameref{lem:FS2-balance-inv} (\ref{lem:FS2-balance-inv}),
and then reactivate $S[k]$ on finishing, so $S[k]$ will proceed
within $O\left(2^{k}\right)$ subsequent delay.

Thus if $X$ is waiting at $S[k]$ due to consecutive sections $S[k..j]$
being deferred, and $S[j+1]$ proceeding, the deferment at $S[k]$
lasts $O\left(\sum_{a=k}^{j+1}2^{a}\right)=O\left(2^{j}\right)$ delay
(by \ref{lem:FS2-balance-inv} again), and $p\cdot2^{j}\le\sqrt{c(m)}\cdot2^{j}\in O\left(c(j)\cdot2^{-j}\right)$
since $2^{2j}\in O\left(\sqrt{c(j)}\right)$. If $S[j]$ had an imbalanced
segment, it would have at least $c(j)\cdot2^{-j}$ stored credits,
and we can use half of it to pay for any needed extra stored credits
at $S[j+1]$ due to the shift. If $S[j+1]$'s buffer had more than
$c(j)$ items, then they carry $c(j)\cdot2^{-j}$ credits, and we
can use half to pay for any needed extra stored credits at $S[j+1]$
and for any credits carried by operations that go on to $S[j+2]$.
In both cases, we can use the other half of those credits to pay for
$p$ times the deferment delay that $X$ takes at $S[k]$.

\uline{Total delay}

There are at most $d$ $\fs_{2}$-calls along $C$, and over all $X$,
each of $b,b',g,i\cdot p^{2},j\cdot p^{2}$ above will sum up to at
most the total number $N$ of $\fs_{2}$-calls, and the total deferment
delay of all $\fs_{2}$-calls along $C$ is $O\left(\frac{N}{p}\right)$.
Therefore the span of $\fs_{2}$ is $O\left(\frac{N}{p}+d\cdot(\log p)^{2}+s_{L}\right)$.
\end{proof}

\section{General Parallel Finger Structures}

\label{sec:GPFS}

To support an arbitrary but fixed number $f$ of movable fingers (besides
the fingers at the ends), while retaining both work-optimality with
respect to the finger bound and good parallelism, we essentially use
a basic parallel finger structure for each \textbf{sector} between
adjacent fingers.

It is easier to do this with $\fs_{1}$, because we are processing
the operations in batches. The finger-move operations are all done
first in a \textbf{finger phase} before the rest of the batch, and
of course we combine finger-move operations on the same finger. Consider
any finger that is between two sectors $R_{0}$ and $R_{1}$. This
finger is sandwiched between the nearest chain $S_{i}$ of $R_{0}$
and the nearest chain $S_{1-i}$ of $R_{1}$. To move this finger
into chain $S_{i}$ of $R_{0}$ past an item in segment $S_{i}[k]$,
we move all the items $I$ between the old and new finger position
from $R_{0}$ to $R_{1}$, roughly as follows:
\begin{enumerate}
\item Cut out the items in $I$ from sector $R_{0}$'s segments $S_{i}[0..k]$
and join them (from small to big) into a single batch $B$.
\item Join the items in sector $R_{1}$'s segments $S_{1-i}[0..k]$ (from
small to big) and shift them into $S_{1-i}[k+1]$ (by a single join).
\item Use $B$ to fill sector $R_{1}$'s sections $S_{1-i}[0..k]$ to target
size except perhaps $S_{1-i}[k]$.
\item Rebalance $R_{0}$ and $R_{1}$ as in $\fs_{1}$'s rebalancing phase
(\step\ref{par:FS1-phases}).
\end{enumerate}
This essentially contributes $O\left(2^{k}\right)$ work and $O(\log n)$
span, because we can preserve the same credit invariant to bound the
rebalancing work and span. It is similar but messier for moving a
finger so far that it goes over the nearer chain of $R_{0}$ and into
its further chain.

After that, we can simply partition the map operations around the
fingers and perform each part on the correct sector in parallel. This
partitioning takes $O(b)$ work and $O(\log b)$ span for each batch
of $b$ operations (see Appendix \ref{sub:par-batch}), and $O(\log b)\wi O\left(\frac{b}{p}+\log p\right)$,
and each sector takes $O(\log n)$ span. Thus we will obtain the desired
work/span bounds (\ref{thm:FS-costs}).

It is much harder for $\fs_{2}$, and considerably complicated, so
we shall not attempt to explain it here.

\section{Work-Stealing Schedulers}

\label{sec:work-steal}

The bounds on the work and span of $\fs_{1}$ and $\fs_{2}$ in \ref{sec:PFS1}
and \ref{sec:PFS2} hold regardless of the scheduler. The performance
bounds for $\fs_{1}$ and $\fs_{2}$ in \ref{sec:main} require a
greedy scheduler, in order to bound the parallel buffer cost. In practice,
we do not have such schedulers. But we can design a suitable work-stealing
scheduler in the QRMW pointer machine model that yields the desired
time bounds (\ref{thm:FS1-perf} and \ref{thm:FS2-perf}) on average,
as we shall explain below.

We make the modest assumption that each processor (in the QRMW pointer
machine) can generate a uniformly random integer in $[1..p]$ and
convert it to a pointer given by a constant lookup-table within $O(1)$
steps. For instance, this can be done if each processor has \textbf{local
RAM} of size $p$ (i.e.~sole access to its own local memory with
$p$ cells and $O(1)$ random access).

The \textbf{\textit{blocking}} work-stealing scheduler in \cite{blumofe1999worksteal}
is for an atomic message passing model, in which multiple concurrent
accesses to each deque are arbitrarily queued and serviced one at
a time. This can be supported by guarding each deque with a CLH lock~\cite{magnusson1994queuelocks},
and the analysis carries over. 

The \textbf{\textit{non-blocking}} work-stealing scheduler in \cite{arora2001worksteal}
assumes $O(1)$ memory contention cost, which is contrary to the QRMW
contention model. Nevertheless, the combinatorial techniques in that
paper can be adapted to prove the desired performance bounds for our
implementation (\ref{def:qrmw-work-stealer}).
\begin{defn}[Non-Blocking Work-Stealing Scheduler]
\label{def:qrmw-work-stealer} The non-blocking work-stealing scheduler
can be implemented in the QRMW pointer machine model as follows:
\begin{itemize}
\item Each processor $i\in[1..p]$ has:

\begin{itemize}
\item A global deque $Q_{i}$ of DAG nodes, shared between \textbf{\textit{owner}}
and \textbf{\textit{stealer}} using Dekker's algorithm.
\item A global non-blocking lock $L_{i}$ (see Appendix \ref{def:try-lock}).
\item A local array $R_{i}[1..p]$ where $R_{i}[j]$ stores a pointer to
$Q_{j}$ and a pointer to $L_{j}$.\quad{}// Used implicitly wherever
needed.
\end{itemize}
\item Each processor $i$ does the following repeatedly:

\begin{block}
\item Access $Q_{i}$ as owner, removing the node $v$ at the bottom if
it is non-empty.
\item If $v$ exists (i.e.~$Q_{i}$ was non-empty):

\begin{block}
\item Execute $v$.
\item Access $Q_{i}$ as owner, inserting all the child nodes generated
by $v$ at the bottom.
\end{block}
\item Otherwise:

\begin{block}
\item Create Int $k$ uniformly randomly chosen from $[1..p]$.
\item If TryLock($L_{k}$):

\begin{block}
\item Access $Q_{k}$ as stealer, removing the node $w$ at the top if it
is non-empty.
\item Unlock($L_{k}$).
\item If $w$ exists (i.e.~$Q_{k}$ was non-empty):

\begin{block}
\item Execute $w$.
\item Access $Q_{i}$ as owner, inserting all the child nodes generated
by $w$ at the bottom.
\end{block}
\end{block}
\end{block}
\end{block}
\end{itemize}
\end{defn}

\section{Conclusions}

This paper presents two parallel finger structures that are work-optimal
with respect to the finger bound, and the faster version has a lower
span by using careful pipelining. Pipelining techniques to reduce
the span of data structure operations have been explored before~\cite{BlellochRe97,OPWM}.
As indicated by our results, the extended implicit batching framework
combines nicely with pipelining and is a promising approach in the
design of parallel data structures. 

Nevertheless, despite the common framework, the parallel finger structures
in this paper and the parallel working-set map in \cite{OPWM} rely
on different ad-hoc techniques and analysis, and it raises the obvious
interesting question of whether there is a way to obtain a batch-parallel
splay tree in the same framework, that satisfies both the working-set
property and the finger property.

\clearpage{}

\renewcommand{\thesection}{\hspace{-1em}}

\section{Appendix}

\renewcommand{\thesection}{A}

Here we spell out the model details, building blocks and supporting
theorems used in our paper.

\subsection{QRMW Pointer Machine Model}

\label{sub:qrmw-ppm}

QRMW stands for \textbf{queued read-modify-write}, as described in
\cite{dwork1997contention}. In this contention model, asynchronous
processors perform memory accesses via read-modify-write (RMW) operations
(including read, write, test-and-set, fetch-and-add, compare-and-swap),
which are supported by almost all modern architectures. Also, to capture
contention costs, multiple memory requests to the same memory cell
are FIFO-queued and serviced one at a time, and the processor making
each memory request is blocked until the request has been serviced.

In the \textbf{parallel pointer machine}, each processor has a fixed
number of local registers and memory accesses are done only via pointers,
which can be locally stored or tested for equality (but no pointer
arithmetic). The QRMW pointer machine model, introduced in \cite{OPWM},
extends the parallel pointer machine model in \cite{goodrich1996parallelsort}
to RMW operations. In this model, each memory node has a fixed number
of memory cells, and each memory cell can hold a single field, which
is either an integer or a pointer. Each processor also has a fixed
number of local registers, each of which can hold a single field.
The basic operations that a processor can perform include arithmetic
operations on integers in its registers, equality-test between pointers
in its registers, creating a new memory node and obtaining a pointer
to it, and RMW operations. An RMW operation can be performed on any
memory cell via a pointer to the memory node that it belongs to.

All operations except for RMW operations take one step each. RMW operations
on each memory cell are FIFO-queued to be serviced, and the first
RMW operation in the queue (if any) is serviced at each time step.
The processor making each memory request is blocked until the request
has been serviced.

\subsection{Parallel Batch Operations}

\label{sub:par-batch}

We rely on the following basic operations on batches:
\begin{itemize}
\item Split a given batch of $n$ items into left and right parts around
a given position, within $O(\log n)$ work/span.
\item Partition a given batch of $n$ items into lower and upper parts around
a given pivot, within $O(n)$ work and $O(\log n)$ span.
\item Partition a sorted batch of $n$ items around a sorted batch of $k$
pivots, within $O(k\cdot\log n)$ work and $O(\log n+\log k)$ span.
\item Join a batch of batches with $n$ total items, within $O(n)$ work
and $O(\log n)$ span.
\item Merge two sorted batches with $n$ total items, optionally combining
duplicates, within $O(n)$ work and $O(\log n)$ span if the combining
procedure takes $O(1)$ work/span.
\end{itemize}
These can be implemented in the QRMW pointer machine model~\cite{P23T}
with each batch stored as a\textbf{ BBT} (leaf-based height-balanced
binary tree with an item at each leaf). They can also be implemented
(more easily) in the binary forking model in \cite{blelloch2019mtramalgo}
with each batch stored in an array. For instance, joining a batch
of arrays can be done by using the standard prefix-sum technique to
compute the total size of the first $k$ arrays, and hence we can
copy each array in parallel into the final output array, and merging
two sorted arrays can be done by the algorithm given in \cite{jaja1992parallelalgo}
(section 2.4) and \cite{nodari2016ics643}.

A related data structure that we also rely on is the bunch data structure,
which is defined as follows.
\begin{defn}[Bunch Structure]
\label{def:bunch} A \textbf{bunch} is an unsorted set supporting
addition of any batch of new elements within $O(1)$ work/span and
conversion to a batch within $O(b)$ work and $O(\log b)$ span if
it has size $b$. A bunch can be implemented using a complete binary
tree with batches at the leaves, with a linked list threaded through
each level to support adding a new batch as a leaf in $O(1)$ work/span.
To convert a bunch to a batch, we treat the bunch as a batch of batches
and parallel join all the batches.
\end{defn}

\subsection{Batch-Parallel Map}

\label{sub:batch-map}

In this paper we rely on a parallel map that supports the following
operations:
\begin{itemize}
\item \textbf{Unsorted batch search:} Search for an unsorted input batch
of $b$ items (not necessarily distinct), tagging each search item
with the result, all within $O(b\cdot\log n)$ work and $O(\log b\cdot\log n)$
span, where $n$ is the map size.
\item \textbf{Sorted batch access:} Perform an item-sorted input batch of
$b$ operations on distinct items, tagging each operation with the
result, all within $O\left(b\cdot\log n\right)$ work and $O(\log b+\log n)$
span, where $n$ is the map size before the batch access.
\item \textbf{Split}: Split a map $M$ of size $k$ around a given pivot
rank $r$ into two maps $M_{1},M_{2}$, where $M_{1}$ contains the
items with ranks at most $r$ in $M$, and $M_{2}$ contains the items
with ranks more than $r$ in $M$, within $O(\log k)$ work/span.
\item \textbf{Join}: Join maps $M_{1},M_{2}$ of total size $k$ where every
item in $M_{1}$ is less than every item in $M_{2}$, within $O(\log k)$
work/span.
\end{itemize}
This can be achieved in the QRMW pointer machine model~\cite{P23T},
and also (more easily) in the binary forking model~\cite{blelloch2019mtramalgo}.

\subsection{Parallel Buffer}

\label{sub:par-buffer}

To facilitate extended implicit batching, we can use any parallel
buffer implementation that takes $O(p+b)$ work and $O(\log p+\log b)$
span per batch of size $b$ (on $p$ processors), any operation that
arrives is (regardless of the scheduler) within $O(1)$ span included
in the batch that is being flushed or in the next batch, and there
are always at most $\frac{1}{2}p+q$ ready buffer nodes (active threads
of the buffer) where $q$ is the number of operations that are currently
buffered or being flushed. This would entail the following parallel
buffer overhead~\cite{OPWM} (and we reproduce the proof here).
\begin{thm}[Parallel Buffer Cost]
\label{rem:par-buff-cost} Take any program $P$ using an implicitly
batched data structure $M$ that is run using any greedy scheduler.
Then the cost (\ref{def:effective}) of the parallel buffer for $M$
is $O\left(\frac{T_{1}+w}{p}+d\cdot\log p\right)$, where $T_{1}$
is the work of all the $P$-nodes, and $w$ is the work taken by $M$,
and $d$ is the maximum number of $M$-calls on any path in the program
DAG $D$.\end{thm}
\begin{proof}
Let $t_{1}$ and $t_{\infty}$ be the total work and span (\ref{def:effective})
respectively of the parallel buffer for $M$. Let $N$ be the total
number of operations on $M$. Consider each batch $B$ of $b$ operations
on $M$. Let $t_{B}$ be span taken by the buffer on $B$. If $b\le p^{2}$,
then $t_{B}\in O(\log p)$. If $b>p^{2}$, then $t_{B}\in O(\log b)\wi O\left(\frac{b}{p}\right)$.
Thus $t_{B}\in O\left(\frac{b}{p}+\log p\right)$ and hence $t_{\infty}\in O\left(\frac{N}{p}+d\cdot\log p\right)$.

Now consider the actual execution of the execution DAG $E$ of the
program $P$ using $M$. At each time step, the buffer is processing
at most two consecutive batches, so we shall analyze the buffer work
done during the time interval for each pair of consecutive batches
$B$ and $B'$, where $B$ has $b$ operations and $B'$ has $b'$
operations.

If $b+b'\ge\frac{1}{6}p$, then the buffer work done on $B$ and $B'$
is $O(b+b')$.

If $b+b'<\frac{1}{6}p$, then there are at most $\frac{1}{2}p+(b+b')<\frac{2}{3}p$
ready buffer nodes in $E$, so at least one of the following holds
at each time step in this interval:
\begin{itemize}
\item At least $\frac{1}{6}p$ ready $P$-nodes in $E$ are being executed.
These steps take at most $O(T_{1})$ work over all intervals.
\item At least $\frac{1}{6}p$ ready $M$-nodes in $E$ are being executed.
These steps take at most $O(w)$ work over all intervals.
\item At most $p$ ready nodes in $E$ are being executed. All ready buffer
nodes in $E$ are being executed (by greedy scheduling), so over all
intervals there are $O(t_{\infty})$ such steps, taking $O(p\cdot t_{\infty})$
work.
\end{itemize}
Therefore $\frac{t_{1}}{p}\in O(\frac{T_{1}+w}{p}+t_{\infty})$, and
hence the buffer's cost is $\frac{t_{1}}{p}+t_{\infty}\in O\left(\frac{T_{1}+w}{p}+d\cdot\log p\right)$
since $N\le T_{1}$.
\end{proof}
The parallel buffer for each data structure $M$ can be implemented
using a static BBT (leaf-based balanced binary tree), with a sub-buffer
at each leaf node, one for each processor, and a flag at each internal
node. Each sub-buffer stores its operations as the leaves of a complete
binary tree with a linked list threaded through each level. Whenever
a thread $\tau$ makes a call to $M$, the processor running $\tau$
suspends it and inserts the call together with a callback (i.e.~a
structure with a pointer to $\tau$ and a field for the result) into
the sub-buffer for that processor. Then the processor walks up the
BBT from leaf to root, test-and-setting each flag along the way, terminating
if it was already set. On reaching the root, the processor notifies
$M$ (by reactivating it), which can decide when to flush the buffer.
$M$ can also query whether the parallel buffer is non-empty, defined
as whether the flag at the root is set. $M$ can eventually return
the result of the call via the callback (i.e.~by updating the result
field and then resuming $\tau$).

Whenever the buffer is flushed (by $M$), all sub-buffers are swapped
out by a parallel recursion on the BBT, replaced by new sub-buffers
in a newly constructed static BBT. We then wait for all pending insertions
into the old sub-buffers to be completed, before joining their contents
into a single batch to be returned (to $M$). To do so, each processor
$i$ has a flag $y_{i}$ initialized to $true$, and a thread field
$\phi_{i}$ initialized to $null$. Whenever it inserts an $M$-call
$X$, it sets $y_{i}:=false$, then inserts $X$ into the (current)
sub-buffer, then resumes $\phi_{i}$ if $\f{TestAndSet}(y_{i})=true$.
To wait for pending insertions into the old sub-buffer for processor
$i$, we store a pointer to the current thread in $\phi_{i}$ and
then suspend it if $\f{TestAndSet}(y_{i})=false$.

Inserting into each sub-buffer can be done in $O(1)$ time. Test-and-setting
each flag in the BBT also takes $O(1)$ time, because at most three
processors ever access it. Each static BBT takes $O(p)$ work and
$O(\log p)$ span to initialize. Each data structure call takes $O(p)$
work and $O(\log p)$ span for a processor to reach the root, because
the flags ensure that only $O(1)$ work is done per node in traversing
the BBT. Joining the contents of the sub-buffers takes $O(p+b)$ work
and $O(\log p+\log b)$ span if the resulting joined batch is of size
$b$. It is also easy to ensure that flushing uses at most $\frac{1}{2}p+b$
threads where $b$ is the size of the flushed batch. Thus this parallel
buffer implementation has the desired properties that support extended
implicit batching.

It is worth noting that the parallel buffer can be implemented in
the dynamic multithreading paradigm, like all other data structures
and algorithms in this paper, but it requires the ability for a thread
to have $O(1)$-time access to the sub-buffer for the processor running
it, so that it can insert each data structure-call into the sub-buffer
in $O(1)$ work/span. This can be done if each processor has a \textbf{local
array} of size $p$ (i.e.~it is accessible only by that processor
but supports $O(1)$ random access) for each implicitly batched data
structure, and each thread can retrieve the id of the processor running
it. But in the QRMW pointer machine model this is not necessary if
the program uses a fixed set of implicitly batched data structures,
since each processor can be initialized with a (constant) pointer
to a structure that always points to the current sub-buffer for that
processor.

\subsection{Sorting Theorems}

The items in the search problem can come from any arbitrary set $S$
that is linearly ordered by a given comparison function, and we shall
assume that $S$ has at least two items. As is standard, let $S^{n}$
be the set of all length-$n$ sequences from $S$. Search structures
can often be adapted to implement sorting algorithms~\footnote{A sorting algorithm is a procedure that given any input sequence will
output a sequence of pointers to the input items in sorted order.}, in which case any lower bound on complexity of sorting typically
implies a lower bound on the costs of the search structure. For the
proofs of \nameref{thm:FS1-work} and \nameref{thm:FS2-work} we need
a crucial lemma that the entropy bound is a lower bound for (comparison-based)
sorting, as precisely stated below.
\begin{lem}[Sorting Entropy Bound]
\label{lem:sort-entr-bound} For any sequence $I$ in $S^{n}$ with
item frequencies $q_{1..u}$ (i.e.~$\sum_{i=1}^{u}q_{i}=n$), any
sorting algorithm requires $\Omega(H)$ comparisons on average over
all (distinct) rearrangements of $I$, where $H=\sum_{i=1}^{u}\left(q_{i}\cdot\log\frac{n}{q_{i}}\right)$
is the entropy of $I$.~\cite{munro1976entropysort}
\end{lem}
From this we immediately get a relation (\ref{thm:max-finger}) between
the entropy bound and the maximum finger bound (i.e.~the maximum
finger bound over all permutations), because we can use a finger-tree
to perform sorting.
\begin{defn}[Finger-Tree Sort]
 Let $FSort$ be the sequential algorithm that sorts an input sequence
$I$ as follows:
\begin{block}
\item Create an empty finger-tree $F$ (with one finger at each end) that
stores linked lists of items. For each item $x$ in $I$, if $F$
already has a linked list of copies of $x$, then append $x$ to that
linked list, otherwise insert a linked list containing just $x$ into
$F$. At the end iterate through $F$ to produce the desired sorted
sequence.
\end{block}
\end{defn}

\begin{defn}[In-order Item Frequencies]
 A sequence $I$ in $S^{n}$ is said to have \textbf{in-order item
frequencies} $q_{1..u}$ if the $i$-th smallest item in $I$ occurs
$q_{i}$ times in $I$.\end{defn}
\begin{thm}[Maximum Finger Bound]
\label{thm:max-finger} Take any sequence $I$ in $S^{n}$ with in-order
item frequencies $q_{1..u}$. Then the \textbf{maximum finger bound}
for $I$, defined as $MF_{I}=\sum_{i=1}^{u}q_{i}\cdot(\log\min(i,u+1-i)+1)$,
satisfies $MF_{I}\in\Omega(H)$ where $H=\sum_{i=1}^{u}\left(q_{i}\cdot\log\frac{n}{q_{i}}\right)$.\end{thm}
\begin{proof}
By the \nameref{lem:sort-entr-bound} (\ref{lem:sort-entr-bound})
let $J$ be a rearrangement of $I$ such that $FSort(J)$ takes $\Omega(H)$
comparisons. Clearly $FSort(J)$ also takes $O\left(MF_{J}\right)=O\left(MF_{I}\right)$
comparisons, and hence $MF_{q}\in\Omega(H)$.
\end{proof}
Finally we give a parallel sorting algorithm $PESort$ that achieves
the entropy bound for work but yet takes only $O\left((\log n)^{2}\right)$
span on a list of $n$ items, which we need in our parallel finger
structure. For comparison, we also give the simpler parallel merge-sort
$PMSort$. The input and output lists are each stored in a batch (leaf-based
balanced binary tree), and these algorithms work in the QRMW pointer
machine model.

We shall use the following notation for every binary tree $T$: $T\a{root}$
is its root, and for each node $v$ of $T$, $v\a{left}$ and $v\a{right}$
are its child nodes, and $v\a{height}$ is the height of the subtree
at $v$, and $v\a{size}$ is the number of leaves of the subtree at
$v$.
\begin{defn}[Parallel Merge-Sort]
\label{def:par-msort} Let $PMSort$ be the procedure that does the
following on an input batch $I$ of items:
\begin{block}
\item If $I\a{size}\le1$, return $I$. Otherwise, compute in parallel $A=PMSort(I\a{left})$
and $B=PMSort(I\a{right})$, and then parallel merge (\ref{sub:par-batch})
$A$ and $B$ into an item-sorted batch $C$, and then return $C$.
\end{block}
\end{defn}
\begin{thm}[$PMSort$ Costs]
 $PMSort$ sorts every sequence $I$ in $S^{n}$ within $O(n\cdot\log n)$
work and $O\left((\log n)^{2}\right)$ span.\end{thm}
\begin{proof}
The claim follows directly from the work/span bounds for parallel
merging (\ref{sub:par-batch}) and $I\a{height}\in O(\log n)$.\end{proof}
\begin{defn}[Parallel Entropy-Sort]
\label{def:par-esort} Define a \textbf{bundle} of an item $x$ to
be a BT (binary tree) in which every leaf has a tagged copy of $x$.
Let $PESort$ be the parallel merge-sort variant that does the following
on an input batch $I$ of items:
\begin{block}
\item If $I\a{size}\le1$, return $I$. Otherwise, compute in parallel $A=PESort(I\a{left})$
and $B=PESort(I\a{right})$, and then parallel merge (\ref{sub:par-batch})
$A$ and $B$ into an item-sorted batch $C$ of bundles, combining
bundles of the same item into one by simply making them the child
subtrees of a new bundle, and then return $C$.
\end{block}
Then $PESort(I)$ returns an item-sorted batch of bundles, with one
bundle (of all the tagged copies) for each distinct item in $I$,
and clearly each bundle has height at most $I\a{height}$.\end{defn}
\begin{thm}[$PESort$ Costs]
\label{thm:par-esort-cost} $PESort$ sorts every sequence $I$ in
$S^{n}$ with item frequencies $q_{1..u}$ within $O(H+n)$ work and
$O\left((\log n)^{2}\right)$ span, where $H=\sum_{i=1}^{u}\left(q_{i}\cdot\ln\frac{n}{q_{i}}\right)$.\end{thm}
\begin{proof}
Consider the merge-tree $T$, in which each node is the result of
parallel merging its child nodes. Note that $T\a{height}=I\a{height}\in O(\log n)$,
and that each item in $I$ occurs in at most one bundle in each node
of $T$. Clearly the work done is $O(1)$ times the total length of
all the parallel merged batches (\ref{sub:par-batch}). Thus the work
done can be divided per item; work done on item $x$ takes $O(1)$
times the number of nodes of $T$ that contain a bundle of $x$, and
there are $O\left(k\cdot\log\frac{n}{k}+k\right)$ such nodes where
$k$ is the frequency of $x$ in $I$, by \ref{lem:subtree-size-bound}
below. Therefore $PESort(I)$ takes $O\left(\sum_{i=1}^{u}\left(q_{i}\cdot\log\frac{n}{q_{i}}+q_{i}\right)\right)\wi O(H+n)$
work. The span bound on $PESort(I)$ is immediate from the span bound
on parallel merging (\ref{sub:par-batch}).
\end{proof}

\begin{lem}[BBT Subtree Size Bound]
\label{lem:subtree-size-bound} Given any BBT $T$ with $n$ leaves
of which $k$ are marked with $k>0$, and with each internal node
marked iff it is on a path from the root to a marked leaf, the number
of marked nodes of $T$ is $O\left(k\cdot\log\frac{n}{k}+k\right)$.\end{lem}
\begin{proof}
We shall iteratively change the set of marked leaves of $T$, and
accordingly update the internal nodes so that each of them is marked
iff it is on a path from the root to a marked leaf. At each step,
if there is a marked node $u$ with a marked child $v$ and an unmarked
child $w$ such that $v$ has two marked children, then unmark the
rightmost marked leaf $x$ in the subtree at $v$ and mark the deepest
leaf $y$ in the subtree at $w$. This will not decrease the number
of marked nodes, because unmarking $x$ results in unmarking at most
$v\a{right}\a{height}$ internal nodes, and marking $y$ results in
marking at least $w\a{height}$ internal nodes, and $v\a{right}\a{height}\le w\a{height}$
since $T$ is a BBT.

Note that each step decreases the sum of the lengths of all the paths
from the root to the marked nodes with two marked children, so this
iterative procedure terminates after finitely many steps. After that,
for every node $v$ with only one marked child, there is only one
marked leaf in the subtree at $v$. Let $A$ be the set of marked
nodes with two marked children, and $B$ be the set of marked nodes
not in $A$ but with a parent in $A$. Then there are exactly $(k-1)$
nodes in $A$, and exactly $k$ nodes in $B$, and the subtrees at
nodes in $B$ are disjoint, so $\sum_{v\in B}v\a{size}\le n$. Since
every marked node is either in $A$ or on the downward path of marked
nodes from some node in $B$, the number of marked nodes is at most
$(k-1)+\sum_{v\in B}(v\a{height}+1)\in O\left(k+\sum_{v\in B}\log v\a{size}\right)\wi O\left(k+k\cdot\log\frac{n}{k}\right)$
by Jensen's inequality.\end{proof}
\begin{rem*}
See \cite{P23T} (Subtree Size Bound) for a generalization of \ref{lem:subtree-size-bound}
with a different proof, but if we want a bound with explicit constants
then the above proof yields a tighter bound for a BBT.
\end{rem*}
$PESort$ is all we need for the parallel finger search structures
$\fs_{1}$ and $\fs_{2}$, but we can in fact obtain a full parallel
entropy-sorting algorithm, namely one that outputs a single item-sorted
batch of all the (tagged copies of) items in the input sequence $I$
from $S^{n}$ and satisfies the entropy bound for work. Specifically,
we can convert each bundle in $PESort(I)$ to a batch (\ref{def:bundle-balance}),
and then parallel join (\ref{sub:par-batch}) all those batches to
obtain the desired output.
\begin{defn}[Bundle Balancing]
\label{def:bundle-balance} A bundle $B$ of size $b$ and height
$h$ is balanced as follows:
\begin{block}
\item Recursively construct a linked list through all the leaves of $B$,
and mark the leaves of $B$ with ($1$-based) rank of the form $(i\cdot h+1)$,
and then extract those marked leaves as a batch $P$ (by parallel
filtering as described in \cite{P23T}). Then at each leaf $v$ in
$P$, construct and store at $v$ a batch of the items in $B$ with
ranks $i\cdot h+1$ to $(i+1)\cdot h$, obtained by traversing the
linked list forward. Now $P$ is essentially a batch of size-$h$
batches (except perhaps the last smaller batch), which we then recursively
join to obtain the batch of all items in $G$ (alternatively, but
less efficiently, simply parallel join $P$).
\end{block}
\end{defn}

\begin{thm}[Bundle Balancing Costs]
 Balancing a bundle $B$ of size $b$ and height $h$ takes $O(b)$
work and $O(h)$ span.\end{thm}
\begin{proof}
Note that $B$ has less internal nodes than leaves, and so constructing
the linked list takes $O(b)$ work and $O(h)$ span. Extracting the
batch $P$ of items of $B$ with ranks at intervals of $h$ takes
$O\left(b+P\a{size}\cdot h\right)=O(b)$ work and $O(h)$ span. Constructing
the batches of items in-between those in $P$ takes $O(b)$ work and
$O(P\a{height}+h)\wi O(h)$ span, and recursively joining them takes
$O(1)$ work and span per node of $P$ (except $O(h)$ span for the
first joining involving the last batch).
\end{proof}

\subsection{Locking Mechanisms\label{sub:locking}}

Here we give pseudo-code implementations of the various locking mechanisms
used as primitives in this paper (\ref{sub:primitives}), which have
the claimed properties under the QRMW memory contention model.

The \textbf{non-blocking lock} is trivially implemented using test-and-set
as shown in TryLock/Unlock below.
\begin{defn}[Non-Blocking Lock]
\label{def:try-lock}~
\begin{block}
\item \textbf{TryLock( Bool $x$ ):}

\begin{block}
\item Return $\neg\f{TestAndSet}(x)$.
\end{block}
\item \textbf{Unlock( Bool $x$ ):}

\begin{block}
\item Set $x:=false$.
\end{block}
\end{block}
\end{defn}
Next is the \textbf{reactivation} wrapper for a procedure $P$, which
can be implemented using fetch-and-add and guarantees the following
according to some linearization~\cite{P23T}:
\begin{enumerate}
\item Whenever $P$ is reactivated, there will be a complete run of $P$
that starts after that reactivation.
\item If $P$ is run only via reactivations, then no runs of $P$ overlap,
and there are at most as many runs of $P$ as reactivations of $P$.
\item If $P$ is reactivated by only $k$ threads at any time, then each
reactivation call $C$ finishes within $O(k)$ span, and some run
of $P$ starts within $O(k)$ span after the start of $C$ or the
end of the last run of $P$ that overlaps $C$.\end{enumerate}
\begin{defn}[Reactivation Wrapper]
\label{def:reactivation}~($P$ is the procedure to be guarded by
the wrapper.)
\begin{block}
\item Private Procedure $P$.
\item Private Int $count:=0$.
\item \textbf{Public Reactivate():}

\begin{block}
\item If $\f{FetchAndAdd}(count,1)=0$:

\begin{block}
\item Fork the following:

\begin{block}
\item Do:

\begin{block}
\item Set $count:=1$.
\item $P()$.
\end{block}
\item While $\f{FetchAndAdd}(count,-1)>1$.
\end{block}
\end{block}
\end{block}
\end{block}
\end{defn}

The \textbf{dedicated lock} with keys $[1..k]$, where threads must
use distinct keys to acquire it, can be implemented using fetch-and-add
as shown below and guarantees the following according to some linearization~\cite{OPWM}:
\begin{enumerate}
\item \textbf{\textit{Mutual exclusion}}: Only one thread can hold the lock
at any point in time; a thread becomes the lock holder when it successfully
acquires the lock, and must release the lock before the next successful
acquisition.
\item \textbf{\textit{Fairness and bounded latency}}: When any thread attempts
to acquire the dedicated lock, it will become a pending holder within
$O(k)$ span, and each pending holder will successfully acquire the
lock after at most $1$ subsequent successful acquisition per key
(if every lock holder eventually releases the lock). And whenever
the lock is released, if there is at least one pending holder then
within $O(k)$ span the lock would be successfully acquired again.\end{enumerate}
\begin{defn}[Dedicated Lock]
\label{def:dedi-lock}~($k$ is the number of keys.)
\begin{block}
\item Private Int $count:=0$.
\item Private Int $last:=0$.
\item Private Array $q[1..k]$ initialized with $null$.
\item \textbf{Public Acquire( Int $i$ ):}

\begin{block}
\item If $\f{FetchAndAdd}(count,1)=0$:

\begin{block}
\item Set $last:=i$.
\item Return.
\end{block}
\item Otherwise:

\begin{block}
\item Write pointer to current thread into $q[i]$.
\item Suspend current thread.
\end{block}
\end{block}
\item \textbf{Public Release():}

\begin{block}
\item If $\f{FetchAndAdd}(count,-1)>1$:

\begin{block}
\item Create Int $j:=last$.
\item Create Pointer $t:=null$.
\item While $t=null$:

\begin{block}
\item Set $j:=j\%k+1$.
\item If $q[j]\ne null$, then swap $t,q[j]$.
\end{block}
\item Set $last:=j$.
\item Resume $t$.
\end{block}
\end{block}
\end{block}
\end{defn}
It is worth mentioning that we can easily replace the array $q[1..k]$
in the above implementation by a cyclic linked list, and use the linked
list nodes instead of integers as the keys.

\setlength{\baselineskip}{1em}

\phantomsection\bibliographystyle{plain}
\addcontentsline{toc}{section}{\refname}\bibliography{OPFS}

\end{document}